\documentclass[a4paper,11pt]{article}
\usepackage[utf8]{inputenc}
\usepackage[margin=1in]{geometry}
\usepackage{authblk}

\usepackage{amsmath}
\usepackage{amsthm}
\usepackage{amssymb}
\usepackage{thmtools}
\usepackage{thm-restate}
\usepackage{enumitem}
\usepackage{dsfont}
\usepackage{comment}
\usepackage[ruled,noend,resetcount,linesnumbered]{algorithm2e}

\usepackage[hidelinks]{hyperref}
\usepackage[capitalise]{cleveref}

\newcommand{\Occ}{\operatorname{Occ}}
\newcommand{\per}{\operatorname{per}}
\newcommand{\HAM}{\mathsf{HD}}

\newcommand{\ED}{\mathsf{ED}}
\newcommand{\IDD}{\mathsf{IDD}}
\newcommand{\dd}{\mathinner{.\,.}}
\newcommand{\ceil}[1]{\lceil #1 \rceil}
\newcommand{\floor}[1]{\lfloor #1 \rfloor}
\newcommand{\sub}{\subseteq}
\newcommand{\sm}{\setminus}

\newcommand{\LCE}{\mathrm{LCE}}
\newcommand{\apLCE}[1]{{\LCE}_{\le #1}}
\newcommand{\bpLCE}{\overline{\LCE}{}}

\newcommand{\Oh}{\mathcal{O}}
\newcommand{\Ohtilde}{\tilde{\Oh}}
\newcommand{\Thetatilde}{\tilde{\Theta}}
\newcommand{\eps}{\varepsilon}

\renewcommand{\Pr}{\mathbb{P}}
\newcommand{\Exp}{\mathbb{E}}

\newtheorem{theorem}{Theorem}[section]
\newtheorem*{theorem*}{Theorem}

\newtheorem{lemma}[theorem]{Lemma}
\newtheorem*{lemma*}{Lemma}
\newtheorem{fact}[theorem]{Fact}
\newtheorem{claim}[theorem]{Claim}

\theoremstyle{definition}
\newtheorem{definition}[theorem]{Definition}

\begin{document}

\SetFuncSty{textsf}

\title{Sublinear-Time Algorithms for Computing \& Embedding\\ Gap Edit Distance}

\renewcommand\Affilfont{\normalsize}

\author[1,2]{Tomasz Kociumaka\thanks{Supported by ISF
  grants no. 1278/16 and 1926/19, by a BSF grant no.
  2018364, and by an ERC grant MPM under the EU's Horizon
  2020 Research and Innovation Programme (agreement no.
  683064).}}
\author[2]{Barna Saha\thanks{Supported by 
NSF grants no. 1652303 (NSF CAREER), 1934846 (NSF HDR TRIPODS), and 1909046, and by an Alfred P. Sloan fellowship award.}}

\affil[1]{Bar-Ilan University, Ramat Gan, Israel}
\affil[ ]{\vspace{-1ex}}
\affil[2]{University of California, Berkeley, USA}
\affil[ ]{\href{mailto:kociumaka@berkeley.edu}
  {\nolinkurl{kociumaka@berkeley.edu}}, \href{mailto:barnas@berkeley.edu}
  {\nolinkurl{barnas@berkeley.edu}}}
\date{\vspace{-1.5cm}}
\maketitle

\begin{abstract}
  In this paper, we design new sublinear-time algorithms for solving the gap edit distance problem and for embedding edit distance to Hamming distance.
  For the gap edit distance problem, we give an $\Ohtilde(\frac{n}{k}+k^2)$-time greedy algorithm that distinguishes between length-$n$ input strings with edit distance at most $k$ and those with edit distance exceeding ${(3k+5)}k$.
  This is an improvement and a simplification upon the result of Goldenberg, Krauthgamer, and Saha [FOCS 2019], where the $k$ vs $\Theta(k^2)$ gap edit distance problem is solved in $\Ohtilde(\frac{n}{k}+k^3)$ time.
  We further generalize our result to solve the $k$ vs $k'$ gap edit distance problem in time $\Ohtilde(\frac{nk}{k'}+k^2+ \frac{k^2}{k'}\sqrt{nk})$, strictly improving upon the previously known bound $\Ohtilde(\frac{nk}{k'}+k^3)$. 
  Finally, we show that if the input strings do not have long highly periodic substrings, then already the $k$ vs $(1+\eps)k$ gap edit distance problem can be solved in sublinear time.
  Specifically, if the strings contain no substring of length $\ell$ with period at most $2k$, then the running time we achieve is $\Ohtilde(\frac{n}{\eps^2 k}+k^2\ell)$.

  We further give the first sublinear-time probabilistic embedding of edit distance to Hamming distance.
  For any parameter $p$, our $\Ohtilde(\frac{n}{p})$-time procedure yields an embedding with distortion $\Oh(kp)$, where $k$ is the edit distance of the original strings.
  Specifically, the Hamming distance of the resultant strings is between $\frac{k-p+1}{p+1}$ and $\Oh(k^2)$ with good probability.
  This generalizes the linear-time embedding of Chakraborty, Goldenberg, and Koucký [STOC 2016], where the resultant Hamming distance is between $\frac k2$ and $\Oh(k^2)$. 
  Our algorithm is based on a random walk over samples, which we believe will find other applications in sublinear-time algorithms. 
 \end{abstract}

 \section{Introduction}
 The edit distance, also known as the Levenshtein distance~\cite{Lev65}, is a basic measure of sequence similarity.
 For two strings $X$ and $Y$, the edit distance $\ED(X,Y)$ is defined as the minimum number of character insertions, deletions and substitutions required to transform $X$ into $Y$.
 A natural dynamic programming computes the edit distance of two strings of total length $n$ in $\Oh(n^2)$ time.
 While running a quadratic-time algorithm is prohibitive for many applications, the Strong Exponential Time Hypothesis (SETH)~\cite{IP01} implies that there is no truly subquadratic-time algorithm that computes edit distance exactly~\cite{BI18}. 
 
 The last two decades have seen a surge of interest in designing fast approximation algorithms for edit distance computation~\cite{AKO10,AN20,AO12,BJKK04,BEKMRRS03,BES06,BEGHS18,BR20,CDGKS18,GKS19,KS20}.
 A breakthrough result of Chakraborty, Das, Goldenberg, Koucký, and Saks provided the first constant-factor approximation of edit distance in truly subquadratic time~\cite{CDGKS18}.
 Nearly a decade earlier, Andoni, Krauthgamer, and Onak showed a polyloga\-rithmic-factor approximation for edit distance in near-linear time~\cite{AKO10}.
 Recently, the result of~\cite{CDGKS18} was improved to a constant-factor approximation in near-linear time:
 initially by Brakensiek and Rubinstein~\cite{BR20} as well as Koucký and Saks~\cite{KS20} for the regime of near-linear edit distance, and then by Andoni and Nosatzki~\cite{AN20} for the general case.
 Efficient algorithms for edit distance have also been developed in other models, such as in the quantum and massively parallel framework~\cite{BEGHS18}, and when independent preprocessing of each string is allowed~\cite{GRS20}.

In this paper, we focus on sublinear-time algorithms for edit distance, the study of which was initiated by Batu, Erg\"un, Kilian, Magen, Rashkodnikova, Rubinfeld, and Sami~\cite{BEKMRRS03}, and then continued in~\cite{AN10,AO12,GKS19,NSS17}. Here, the goal is to distinguish, in time sublinear in $n$, whether the edit distance is at most $k$ or strictly above $k'$ for some $k'\ge k$. This is known as the ($k$~vs~$k'$) \emph{gap edit distance problem}. In computational biology, before an in-depth comparison of new sequences is performed, a quick check to eliminate sequences that are not highly similar can save a significant amount of resources~\cite{DKFPOS}. In text corpora, a super-fast detection of plagiarism upon arrival of a new document can save both time and space. In these applications, $k$ is relatively small and a sublinear-time algorithm with $k'$ relatively close to $k$ could be very useful.

\subsubsection*{Results on Gap Edit Distance}
The algorithm of Batu et al.~\cite{BEKMRRS03} distinguishes between $k=n^{1-\Omega(1)}$ and $k'=\Omega(n)$ in time $\Oh(\max(\frac{k^2}{n},\sqrt{k}))$.
However, their algorithm crucially depends on $k'=\Omega(n)$ and cannot distinguish between, say, $k=n^{0.01}$ and $k'=n^{0.99}$.
A more recent algorithm by Andoni and Onak~\cite{AO12} resolves this issue and can distinguish between $k$ and $k'\ge k\cdot n^{\Omega(1)}$ in time $\Oh(\frac{n^{2+o(1)}k}{(k')^2})$. However, if we want to distinguish between $k$ and $k'=\Theta(k^2)$, then the algorithm of~\cite{AO12} achieves sublinear time only when $k=\omega(n^{1/3})$.
(Setting $k'=\Theta(k^2)$ yields a natural test case for the gap edit distance problem since the best that one can currently distinguish in linear time is $k$ vs $\Theta(k^2)$~\cite{LV88}.)
In a recent work, Goldenberg, Krauthgamer, and Saha~\cite{GKS19} gave an algorithm solving \emph{quadratic gap edit distance problem} in $\Ohtilde(\frac{n}{k}+k^3)$ time\footnote{The $\Ohtilde$ notation hides factors polylogarithmic in the input size and, in case of Monte Carlo randomized algorithms, in the inverse error probability.}, thereby providing a truly sublinear-time algorithm as long as $n^{\Omega(1)} \le k\le n^{1/3-\Omega(1)}$.

Bar-Yossef, Jayram, Krauthgamer, and Kumar~\cite{BJKK04} introduced the gap edit distance problem and solved the quadratic gap edit distance problem for non-repetitive strings.
Their algorithm computes a constant-size sketch
but still requires a linear-time pass over the data.
This result was later generalized to arbitrary sequences~\cite{CGK16} via embedding edit distance into Hamming distance, but again in linear time.
Nevertheless, already the algorithm of Landau and Vishkin~\cite{LV88} computes the edit distance exactly in $\Oh(n+k^2)$ time, and thus also solves the quadratic gap edit distance problem in linear time.
Given the prior works, Goldenberg et al.~\cite{GKS19} raised a question whether it is possible to solve the quadratic gap edit problem in sublinear time for all $k \ge n^{\Omega(1)}$. 
In particular, the running times of the algorithms of Goldenberg et al.~\cite{GKS19} and Andoni and Onak~\cite{AO12} algorithms meet at $k \approx n^{1/3}$, when they become nearly-linear.
In light of the $\Oh(n+k^2)$-time exact algorithm~\cite{LV88}, the presence of a $k^3$ term in the time complexity of~\cite{GKS19} is undesirable, and it is natural to ask if the dependency can be reduced.
In particular, if the polynomial dependency on $k$ can be reduced to $k^2$, then, for $k=O(n^{1/3})$, the contribution of that term is negligible compared to $\frac{n}{k}$.

\paragraph*{Quadratic Gap Edit Distance} We give a simple greedy algorithm solving the quadratic gap edit distance problem in $\Ohtilde(\frac{n}{k}+k^2)$ time. This resolves an open question posed in~\cite{GKS19} as to whether a sublinear-time algorithm for the quadratic gap edit distance is possible for $k=n^{1/3}$. 
Our algorithm improves upon the main result of~\cite{GKS19}, also providing a conceptual simplification.
\paragraph*{$k$ vs $k'$ Gap Edit Distance}
Combining the greedy approach with the structure of computations in~\cite{LV88}, we can solve the $k$ vs $k'$ gap edit distance problem in $\Ohtilde(\frac{nk}{k'}+k^2+\frac{k^2}{k'}\sqrt{kn})$ time. 
For all values of $k'$ and $k$, this is at least as fast as the $\Ohtilde(\frac{nk}{k'}+k^3)$ time bound of~\cite{GKS19}.
\paragraph*{$k$ vs $(1+\eps)k$ Gap Edit Distance}
We can distinguish edit distance at most $k$ and at least $(1+\eps)k$ in $\Ohtilde(\frac{n}{\eps^2 k}+\ell k^2)$ time as long as there is no length-$\ell$ substring with period at most~$2k$. 
Previously, sublinear-time algorithms for distinguishing $k$ vs $(1+\eps)k$ were only known for the very special case of Ulam distance, where each character appears at most once in each string~\cite{AN10,NSS17}.
Note that not only we can allow character repetition, but we get an $({1+\eps})$-approximation as long as the same repetitive structure does not continue for more than $\ell$ consecutive positions or has shortest period larger than $2k$.
This is the case with most text corpora and for biological sequences with interspersed repeats.

\subsubsection*{Embedding Edit Distance to Hamming Distance} 
Along with designing fast approximation algorithms for edit distance, a parallel line of works investigated how edit distance can be embedded into other metric spaces, especially to the Hamming space~\cite{ADGIR03,BES06,CGK16,CGKK18,OR07}. 
Indeed, such embedding results have led to new approximation algorithms for edit distance (e.g., the embedding of~\cite{BES06,OR07} applied in~\cite{AO12,BES06}), as well as new streaming algorithms and document exchange protocols (e.g., the embedding of~\cite{CGK16} applied in~\cite{BZ16,CGK16}). 
In particular, Chakraborty, Goldenberg, and Koucký~\cite{CGK16} provided a probabilistic embedding of edit distance to Hamming distance with linear distortion.
Their algorithm runs in linear time, and if the edit distance between two input sequences is $k$, then the Hamming distance between the resultant sequences is between $\frac k2$ and $\Oh(k^2)$ with good probability.
The embedding is based on performing an interesting one-dimensional random walk which had also been used previously to design fast approximation algorithms for a more general \emph{language edit distance} problem~\cite{S14}.
So far, we are not aware of any sublinear-time metric embedding algorithm from edit distance to Hamming distance. 
In this paper, we design one such algorithm.

\paragraph*{Random Walk over Samples} We show that it is possible to perform a random walk similar to~\cite{S14,CGK16} over a suitably crafted sequence of samples. This leads to the first sublinear-time algorithm for embedding edit distance to Hamming distance: Given any parameter $p=\Omega(\log n)$, our embedding algorithm processes any length-$n$ string in $\Ohtilde(\frac{n}{p})$ time and guarantees that (with good probability) the Hamming distance of the resultant strings is between $\frac{k-p+1}{p+1}$ and $\Oh(k^2)$, where $k$ is the edit distance of the input strings. 
That is, we maintain the same expansion rate as~\cite{CGK16} and allow additional contraction by a factor roughly $p$. 
As the algorithm of~\cite{CGK16} has been very influential (see its applications in~\cite{BZ16,BGZ17,H19,ZZ17}), we believe the technique of random walk over samples will also find other usages in designing sublinear-time and streaming algorithms.

\subsection*{Technical Overview}
The classic Landau--Vishkin algorithm~\cite{LV88} tests whether $\ED(X,Y)\le k$ in $\Oh(n+k^2)$ time, where $n=|X|+|Y|$.
The algorithm fills a dynamic-programming table with cells $d_{i,j}$ for $i\in[0\dd k]$ and $j\in [-k\dd k]$,\footnote{For $\ell,r\in \mathbb{Z}$, we denote $[\ell\dd r)=\{j\in \mathbb{Z}: \ell \le j < r\}$ and $[\ell\dd r]=\{j\in \mathbb{Z}: \ell \le j \le r\}$.} aiming at  $d_{i,j}=\max\{x : \ED(X[0\dd x), Y[0\dd x+j)) \leq i\}$.
In terms of the table of all distances $\ED(X[0\dd x), Y[0\dd y))$, each value $d_{i,j}$ can be interpreted as (the row of) the farthest cell on the $j$th \emph{diagonal} with value $i$ or less. 
After preprocessing $X$ and $Y$ in linear time, the cells $d_{i,j}$ can be filled in $\Oh(1)$ time each, which results in $\Oh(n+k^2)$ time in total.

Our algorithm for the quadratic gap edit distance problem follows the basic framework of~\cite{LV88}. However, instead of computing $\Theta(k^2)$ values $d_{i,j}$, it only computes $\Theta(k)$ values $d_i$ with $i\in [0\dd k]$.
Here, $d_i$ can be interpreted as a relaxed version of $\max_{j=-k}^{k} d_{i,j}$, allowing for a factor-$\Oh(k)$ underestimation of the number of edits $i$.
Now, it suffices to test whether $d_k=|X|$, because
$d_k = |X|$ holds if $\ED(X,Y)\le k$ and, conversely, $\ED(X,Y)=\Oh(k^2)$ holds if $d_k=|X|$.
In addition to uniform sampling at a rate of $\Ohtilde(\frac{1}{k+1})$, identifying each of these $\Oh(k)$ values $d_i$ requires reading $\Ohtilde(k)$ extra characters. This yields a total running time of $\Ohtilde(\frac{n}{k}+k^2)$.

Our algorithm not only improves upon the main result of Goldenberg et al.~\cite{GKS19}, but also provides a conceptual simplification.
Indeed, the algorithm of~\cite{GKS19} also utilizes the high-level structure of~\cite{LV88}, but it identifies all $\Theta(k^2)$ values $d_{i,j}$ (relaxed to allow for factor-$\Oh(k)$ underestimations), paying extra $\Thetatilde(k)$ time per each value.
In order to do so, the algorithm follows a more complex row-by-row approach of an online version~\cite{LMS98} of the Landau--Vishkin algorithm. 

For the general $k$ vs $\Theta(k')$ gap edit distance problem, greedily computing the values $d_i$ for $i\in [0\dd k]$ is sufficient only for $k' = \Omega(k^2)$, when the time complexity becomes $\Ohtilde(\frac{nk}{k'}+k^2)$ if we simply decrease the sampling rate to $\Ohtilde(\frac{k}{k'})$.
Otherwise, each shift between diagonals $j\in [-k\dd k]$, which happens for each $i\in [1\dd k]$ as we determine $d_{i}$ based on $d_{i-1}$, involves up to $2k$ insertions or deletions, whereas to distinguish edit distance $k$ and $\Theta(k')$, we would like to approximate the number of edits within a factor $\Oh(\frac{k'}{k})$.
In order to do so, we decompose the entire set of $2k+1$ diagonals into $\Theta(\frac{k^2}{k'})$ groups of  $\Theta(\frac{k'}{k})$ consecutive diagonals.
Within each of these \emph{wide diagonals}, we compute the (relaxed) maxima of $d_{i,j}$ following our greedy algorithm. 
This approximates the true maxima up to a factor-$\Oh(\frac{k'}{k})$ underestimation of the number of edits.
Computing each of the $\Oh(k)$ values for each wide diagonal requires reading $\Ohtilde(\frac{k'}{k})$ extra characters. 
Hence, the running time is $\Ohtilde(\frac{nk}{k'}+k')$ per wide diagonal and $\Ohtilde({\frac{nk^3}{(k')^2}+k^2})$ in total (across the $\Theta(\frac{k^2}{k'})$ wide diagonals). 
This bound is incomparable to $\Ohtilde(\frac{nk}{k'}+k^3)$ of~\cite{GKS19}: While we pay less on the second term, the uniform sampling rate increases.
In order to decrease the first term, instead of sampling over each wide diagonal independently, we provide a synchronization mechanism so that the global uniform sampling rate remains $\Ohtilde(\frac{k}{k'})$.
This leads to an implementation with running time $\Ohtilde(\frac{nk}{k'}+\frac{k^4}{k'})$, which already improves upon~\cite{GKS19}.
However, synchronizing only over appropriate smaller groups of wide diagonals, we can achieve the running time of $\Ohtilde(\frac{nk}{k'}+k^2+\frac{k^2}{k'}\sqrt{kn})$, which subsumes both $\Ohtilde(\frac{nk^3}{(k')^2}+k^2)$ and $\Ohtilde(\frac{nk}{k'}+\frac{k^4}{k'})$.

Our algorithm for the $k$ vs $(1+\eps)k$ gap edit distance for strings without length-$\ell$ substrings
with period at most $2k$ follows a very different approach, inspired by the existing 
solutions for estimating the Ulam distance~\cite{AN10,NSS17}.
This method consists of three ingredients. 
First, we construct decompositions $X=X_0\cdots X_m$ and $Y=Y_0\cdots Y_m$ into 
phrases of length $\Oh(\ell k)$ such that, if $\ED(X,Y)\le k$, then $\sum_{i=0}^m \ED(X_i,Y_i)\le \ED(X,Y)$
holds with good probability (note that $\ED(X,Y)\le  \sum_{i=0}^m \ED(X_i,Y_i)$ is always true).
The second ingredient estimates $\ED(X_i,Y_i)$ for any given $i$.
This subroutine is then applied for a random sample of indices $i$ by the third ingredient, which distinguishes between $\sum_{i=0}^m \ED(X_i,Y_i)\le k$ and $\sum_{i=0}^m \ED(X_i,Y_i)>(1+\eps)k$ relying on the Chernoff bound.
The assumption that $X$ does not contain long periodic substrings is needed only in the first step:
it lets us uniquely determine the beginning of the phrase $Y_i$ assuming that the initial $\ell$ positions of the phrase $X_i$ are aligned without mismatches in the optimal edit distance alignment (which is true with good probability for a random decomposition of $X$).
We did not optimize the $\Ohtilde(\ell k^2)$ term in our running time $\Ohtilde(\frac{n}{\eps^2k}+\ell k^2)$ to keep our implementation of the other two ingredients much simpler than their counterparts in~\cite{AN10,NSS17}.

A simple \emph{random deletion process}, introduced in~\cite{S14}, solves the quadratic gap edit distance problem in linear time. The algorithm simultaneously scans $X$ and $Y$ from left to right.
If the two currently processed characters $X[x]$ and $Y[y]$ match, they are aligned, and the algorithm proceeds to $X[x+1]$ and $Y[y+1]$.
Otherwise, one of the characters is deleted uniformly at random (that is, the algorithm proceeds to $X[x]$ and $Y[y+1]$ or to $X[x+1]$ and $Y[y]$).
This process can be interpreted as a one-dimensional random walk, and the hitting time of the random walk provides the necessary upper bound on the edit distance. 
In order to conduct a similar process in sublinear query complexity,
we compare $X[x]$ and $Y[y]$ with probability $\Ohtilde(\frac{1}{p})$ only;
otherwise, we simply align $X[x]$ and $Y[y]$. 
We show that performing this random walk over samples is sufficient for the  $k$ vs $\Theta(k^2p)$ gap edit distance problem. 
Finally, we observe that this \emph{random walk over samples} can be implemented in $\Ohtilde(\frac{n}{p})$ time
by batching iterations.

In order to derive an embedding, we modify the random deletion process so that, after learning that $X[x]=Y[y]$,
the algorithm uniformly at random chooses to stay at $X[x]$ and $Y[y]$ or move to $X[x+1]$ and $Y[y+1]$.
This has no impact on the final outcome, but the decision whether the algorithm stays at $X[x]$
or moves to $X[x+1]$ can now be made independently of $Y$.
This allows for an embedding whose shared randomness consists in the set $S\sub [1\dd 3n]$
of iterations $i$ when $X[x]$ is accessed and, for each $i\in S$, a random function
$h_i : \Sigma\to \{0,1\}$. For each iteration $i\in S$, the embedding outputs $X[x]$
and proceeds to $X[x+h_i(X[x])]$. If $Y$ is processed using the same 
shared randomness, the two output strings are, with good probability, at Hamming distance between $\frac{k-p+1}{p+1}$ and $\Oh(k^2)$.

\subsubsection*{Organization}
After introducing the main notations in Section~\ref{sec:prelim}, we describe and analyze our algorithm for the quadratic gap edit distance problem in Section~\ref{sec:simple}.
In Section~\ref{sec:edLCE}, we solve the more general $k$ vs $k'$ gap edit distance problem for $k' = \Oh(k^2)$. 
The $k$ vs $(1+\eps)k$ gap edit distance problem for strings without long periodic substrings is addressed in Section~\ref{sec:PTAS}.
The random walk over samples process is presented in Section~\ref{sec:randomwalk}.
Finally, the embedding result is provided in Section~\ref{sec:embed}.

\subsubsection*{Further Remarks} 
A recent independent work~\cite{BCR20} uses a greedy algorithm similar to ours and achieves a running time of $\tilde{O}(\frac{n}{\sqrt{k}})$ for the quadratic gap edit distance problem. 
This is in contrast to our bound of $\Ohtilde(\frac{n}{k}+k^2)$, which is superior for $k \le n^{2/5}$. At the same time, for $k \ge n^{2/5+o(1)}$, the algorithm of Andoni and Onak~\cite{AO12} has a better running time $\Oh\big(\frac{n^{2+o(1)}}{k^3}\big)$. 
We also remark here that there exists an even simpler algorithm (by now folklore) 
that has query complexity  $\Ohtilde(\frac{n}{\sqrt{k}})$. 
The algorithm samples both sequences $X$ and $Y$ independently with probability $\tilde{\Theta}(\frac{1}{\sqrt{k}})$ so that  $\Pr[ X[x] \text{ and }Y[y]\text{ are sampled}]=\tilde{\Theta}(\frac{1}{k})$ for all $x \in [0\dd |X|)$ and $y \in [0\dd |Y|)$. 
Then, by running the Landau--Vishkin algorithm~\cite{LV88} suitably over the sampled sequences, one can solve the quadratic gap edit distance problem.
Still, it remains open to tightly characterize the time and query complexity of the quadratic gap edit distance problem.

\section{Preliminaries}\label{sec:prelim}
  A \emph{string} $X$ is a finite sequence of characters from an \emph{alphabet} $\Sigma$.
  The length of $X$ is denoted by $|X|$ and, for $i\in [0\dd |X|)$,
  the $i$th character of $X$ is denoted by $X[i]$.
  A string $Y$ is a \emph{substring} of a string $X$ if $Y=X[\ell]X[\ell+1]\cdots X[r-1]$
  for some $0\le \ell \le r \le |X|$. 
  We then say that $Y$ \emph{occurs} in $X$ at position $\ell$.
  The set of positions where $Y$ occurs in $X$ is denoted $\Occ(Y,X)$.
  The \emph{occurrence} of $Y$ at position $\ell$ in $X$ is denoted by $X[\ell \dd r)$ or $X[\ell\dd r-1]$.
  Such an occurrence is a \emph{fragment} of $X$, and it can be represented by (a pointer to) $X$
  and a pair of indices $\ell\le r$. 
  Two fragments (perhaps of different strings) \emph{match} if they are occurrences of the same substring.
  A fragment $X[\ell\dd r)$ is a \emph{prefix} of $X$ if $\ell=0$ and a \emph{suffix} of $X$ if $r=|X|$.

  A positive integer $p$ is a \emph{period} of a string $X$ if $X[i]=X[i+p]$ holds for each $i\in [0\dd |X|-p)$.
  We define $\per(X)$ to be the smallest period of $X$.
  The following result relates periods to occurrences:
  \begin{fact}[{Breslaurer and Galil~\cite[Lemma 3.2]{BG95}}]\label{fct:period}
    If strings $P,T$ satisfy $|T|\le \frac32|P|$, then $\Occ(P,T)$
    forms an arithmetic progression with difference $\per(P)$.
  \end{fact}

  \paragraph*{Hamming distance and edit distance}
  The \emph{Hamming distance} between two strings $X,Y$ of the same length is defined
  as the number of mismatches. Formally, $\HAM(X,Y)=|\{i\in [0\dd |X|) : X[i]\ne Y[i]\}|$.
  The edit distance between two strings $X$ and $Y$ is denoted $\ED(X,Y)$.

  \paragraph*{LCE queries}
  Let $X,Y$ be strings and let $k$ be a non-negative integer.
  For $x\in [0\dd |X|]$ and $y\in [0\dd |Y|]$, we define
  $\LCE_k^{X,Y}(x,y)$ as the largest integer $\ell$ such that
  $\HAM(X[x\dd x+\ell),Y[y\dd y+\ell)) \le k$ (in particular,
  $\ell \le \min(|X|-x,|Y|-y)$ so that $X[x\dd x+\ell)$ and $Y[y\dd y+\ell)$ are well-defined).
  We also set $\LCE_k^{X,Y}(x,y)=0$ if $x\notin [0\dd |X|]$ or $y\notin [0\dd |Y|]$.

  Our algorithms rely on two notions of \emph{approximate} $\LCE$ queries.
  The first variant is sufficient for distinguishing between $\ED(X,Y)\le k$ and $\ED(X,Y)>k(3k+5)$
  in $\Ohtilde(\frac{n}{k+1}+k^2)$ time, while a more general algorithm distinguishing
  between $\ED(X,Y)\le k$ and $\ED(X,Y)>\alpha k$ is based on the more subtle second variant.
  \begin{definition}
    Let $X,Y$ be strings and let $k\ge 0$ be an integer.
    For integers $x,y$, we set $\apLCE{k}^{X,Y}(x,y)$ as any value satisfying
    $\LCE_0^{X,Y}(x,y) \le \apLCE{k}^{X,Y}(x,y) \le \LCE_k^{X,Y}(x,y)$.
  \end{definition}

  \begin{definition}
    Let $X,Y$ be strings and let $r > 0$ be a real parameter.
    For integers $x,y$, we set
    $\bpLCE_r^{X,Y}(x,y)$ as any \emph{random variable} satisfying the following conditions:
    \begin{itemize}
      \item $\bpLCE_r^{X,Y}(x,y) \ge \LCE_0^{X,Y}(x,y)$,
      \item $\Pr\big[\bpLCE_r^{X,Y}(x,y) > \LCE_k^{X,Y}(x,y)\big] \le \exp(-\frac{k+1}{r})$ for every integer $k\ge 0$.
    \end{itemize}
    \end{definition}
    Note that $\bpLCE_r^{X,Y}(x,y)$ for $r = \frac{k+1}{\ln N}$
    satisfies the conditions on $\apLCE{k}^{X,Y}(x,y)$ with probability~${1-\frac{1}{N}}$.
    Thus, $\bpLCE_r$ queries with sufficiently small $r=\Thetatilde(k+1)$ yield $\apLCE{k}$ queries with high probability.

\section{Quadratic Gap Edit Distance}\label{sec:simple}
The classic Landau--Vishkin exact algorithm~\cite{LV88} for testing if $\ED(X,Y)\le k$ 
is given below as \cref{alg:lv}.
The key property of this algorithm is that $d_{i,j}= \max\{x : \ED(X[0\dd x),Y[0\dd x+j))\le i\}$
holds for each $i\in [0\dd k]$ and $j\in [-k\dd k]$.
Since $\LCE_0^{X,Y}$ queries can be answered in $\Oh(1)$ time after linear-time preprocessing,
the running time is $\Oh(|X|+k^2)$.

\begin{algorithm}[h]
  \lForEach{$i \in [0\dd k]$ \KwSty{and} $j \in [-k-1\dd k+1]$}{$d'_{i,j}:=d_{i,j}:=-\infty$}
  $d'_{0,0}:= 0$\;
  \For{$i := 0$ \KwSty{to} $k$}{
    \For{$j := -k$ \KwSty{to} $k$}{
      \lIf{$d'_{i,j}\ne -\infty$}{%
        $d_{i,j} := d'_{i,j}+\LCE_{0}^{X,Y}(d'_{i,j},d'_{i,j}+j)$
      }%
    }  
    \lFor{$j := -k$ \KwSty{to} $k$}{%
      $d'_{i+1,j} := \min(|X|,\max(d_{i,j-1},d_{i,j}+1,d_{i,j+1}+1))$
    }
  }
  \lIf{$||X|-|Y|| \le k$ \KwSty{and} $d_{k,|Y|-|X|} = |X|$}{\Return{YES}}
  \lElse{\Return{NO}}
  \caption{The Landau--Vishkin algorithm~\cite{LV88}}\label{alg:lv}
  \end{algorithm}

The main idea behind the algorithm of Goldenberg et al.~\cite{GKS19}
is that if $\LCE_0$ queries are replaced with $\apLCE{k}$ queries,
then the algorithm is still guaranteed to return YES if $\ED(X,Y)\le k$
and NO if $\ED(X,Y)> k(k+2)$.
The cost of their algorithm is $\Ohtilde(\frac{1}{k+1}|X|)$ plus $\Ohtilde(k)$ per $\apLCE{k}$ query,
which yields $\Ohtilde(\frac{1}{k+1}|X|+k^3)$ in total.
Nevertheless, their implementation is tailored to the specific structure of $\LCE$ queries in \cref{alg:lv},
and it requires these queries to be asked and answered in a certain order,
which makes them use an online variant~\cite{LMS98} of the Landau--Vishkin algorithm.

An auxiliary result of this paper is that $\apLCE{k}^{X,Y}(x,y)$ queries with $|x-y|\le k$
can be answered in $\Ohtilde(k)$ time after $\Ohtilde(\frac{1}{k+1}|X|)$ preprocessing,
which immediately yields a more modular implementation of the algorithm of~\cite{GKS19}.
In fact, we show that $\Ohtilde(k)$ time is sufficient to answer
\emph{all} queries $\apLCE{k}^{X,Y}(x,y)$ with fixed $x$ and arbitrary $y\in [x-k\dd x+k]$.

Unfortunately, this does not give a direct speed-up, because the values $d'_{i,j}$ in \cref{alg:lv} might all be different.
However, given that relaxing $\LCE_0$ queries to $\apLCE{k}$ queries yields a cost of up to $k$ mismatches for every $\apLCE{k}^{X,Y}(x,y)$ query,
the algorithm may as well pay $\Oh(k)$ further edits (insertions or deletions) to change the \emph{shift}
$j=y-x$ arbitrarily. 
As a result, we do not need to consider each shift $j$ separately.
This results in a much simpler \cref{alg:simple}.

\begin{algorithm}
  \caption{Simple algorithm}\label{alg:simple}
  $d'_0 := 0$\;
  \For{$i := 0$ \KwSty{to} $k$}{
    $d_i := d'_i + \max_{\delta=-k}^k \apLCE{k}^{X,Y}(d'_i,d'_i+\delta)$\;
    $d'_{i+1} := \min(|X|,d_i + 1)$\;
  }
  \lIf{$||X|-|Y|| \le k$ \KwSty{and} $d_k = |X|$}{\Return{YES}}
  \lElse{\Return{NO}}
\end{algorithm}

\begin{lemma}\label{lem:simple_correct}
\cref{alg:simple} returns YES if $\ED(X,Y)\le k$ and NO if $\ED(X,Y)>(3k+5)k$.
\end{lemma}
\begin{proof}
We prove two claims on the values $d'_i$ and $d_i$.
\begin{claim}\label{claim:simple_no}
  Each $i\in [0\dd k]$ has the following properties:
  \begin{enumerate}[label={\rm(\alph*)}]
    \item\label{it:simple_di} $\ED(X[0\dd d'_i),Y[0\dd y)) \le (3k+1)i+k$ for every $y\in [d'_i-k\dd d'_i+k]\cap [0\dd |Y|]$;
    \item\label{it:simple_dip} $\ED(X[0\dd d_i),Y[0\dd y)) \le (3k+1)i+4k$ for every $y\in [d_i-k\dd d_i+k]\cap [0\dd |Y|]$.
  \end{enumerate}
\end{claim}
\begin{proof}
We proceed by induction on $i$. Our base case is Property~\ref{it:simple_di} for $i=0$.
Since $d'_0 = 0$,
we have $\ED(X[0\dd d'_0),\allowbreak Y[0\dd y))=y \le k$ for $y\in [d'_0-k\dd \allowbreak d'_0+k]\cap [0\dd |Y|]$.

Next, we shall prove that Property~\ref{it:simple_dip} holds for $i\ge 0$ assuming that Property~\ref{it:simple_di}
is true for $i$.
By definition of $\apLCE{k}$ queries, we have $d_i \le d'_i+\LCE_k^{X,Y}(d'_i,y')$ 
for some position $y'\in [d'_i-k\dd d'_i+k]\cap [0\dd |Y|]$,
and thus $\HAM(X[d'_i\dd d_i),Y[y'\dd y'+d_i-d'_i))\le k$.
The assumption yields
$\ED(X[0\dd d'_i),Y[0\dd y'))\le (3k+1)i+k$,
so we have $\ED(X[0\dd d_i),Y[0\dd y'+d_i-d'_i))\le (3k+1)i+2k$.
Due to $|y'+d_i-d'_i-y|\le 2k$, we conclude that 
$\ED(X[0\dd d_i),Y[0\dd y))\le (3k+1)i+4k$.

Finally, we shall prove that Property~\ref{it:simple_di} holds for $i>0$ assuming that Property~\ref{it:simple_dip}
is true for $i-1$.
Since $d'_i \le d_{i-1}+1$, the assumption yields 
$\ED(X[0\dd {d'_i-1}),\allowbreak Y[0\dd {y-1}))\le {(3k+1)}{(i-1)}+4k$,
and thus 
$\ED(X[0\dd d'_i),\allowbreak Y[0\dd y))\le 1+(3k+1)(i-1)+4k = (3k+1)i+k$.
\end{proof}
Thus, $\ED(X,Y){\le} (3k+5)k$ if the algorithm returns YES.

\begin{claim}\label{clm:simple_yes}
If $\ED(X[0\dd x),Y[0\dd y)) = i\in [0\dd k]$ for $x\in [0\dd |X|]$ and $y\in [0\dd |Y|]$, then $x \le d_i$.
\end{claim}
\begin{proof}
We proceed by induction on $i$. Both in the base case of $i=0$ and the inductive step of $i>0$,
we shall prove that $x \le d'_i+\max_{\delta=-k}^k \LCE_0^{X,Y}(d'_i, d'_i+\delta)$.
Since $d_i \ge d'_i +\max_{\delta=-k}^k \LCE_0^{X,Y}(d'_i, d'_i+j)$ holds by definition 
of $\apLCE{k}$ queries, this implies the claim.

In the base case of $i=0$, we have $X[0\dd x)=Y[0\dd y)$ and $d'_0=0$. Consequently, 
$x \le \LCE_0^{X,Y}(0,0) \le d'_0+\max_{\delta=-k}^k \LCE_0^{X,Y}(d'_0, d'_0+\delta)$.

For $i>0$, we consider an optimal alignment between $X[0\dd x)$ and $Y[0\dd y)$, and we distinguish its maximum prefix with $i-1$ edits.
This yields positions $x',x''\in [0\dd x]$ and $y',y''\in [0\dd y]$ with $x''-x'\in \{0,1\}$ and $y''-y'\in \{0,1\}$ such that $\ED(X[0\dd x'),Y[0\dd y'))=i-1$ and $X[x''\dd x)=Y[y''\dd y)$.
The inductive assumption yields $x' \le d_{i-1}$, which implies $x'' \le \min(x,{d_{i-1}+1})\le  d'_i$.
Due to $X[x''\dd x)=Y[y''\dd y)$, we have $\LCE_0^{X,Y}(x'',y'')\ge x-x''$.
By $x'' \le d'_i$, this implies $\LCE_0^{X,Y}(d'_i, {d'_i+y-x})\ge x-d'_i$. 
Since $|y-x|\le k$, we conclude that
$x = d'_i + (x-d'_i) \le d'_i + \LCE_0^{X,Y}(d'_i, d'_i+y-x) \le d'_i+\max_{\delta=-k}^k \LCE_0^{X,Y}(d'_i, d'_i+\delta)$.
\end{proof}
Hence, the algorithm returns YES if $\ED(X,Y)\le k$.
\end{proof}

A data structure computing $\apLCE{k}^{X,Y}(x,y)$ for a given $x$ and all $y\in [x-k\dd x+k]$ is complicated, but a simpler result stated below
and proved in \cref{sec:apLCE} suffices here. 

\begin{restatable}{proposition}{aplce}\label{prp:aplce}
  There exists an algorithm that, given strings $X$ and $Y$, an integer $k\ge 0$, an index $i$, and a range of indices $J$,
  computes $\ell:=\max_{j\in J}\apLCE{k}^{X,Y}(i,j)$.
  With high probability, the algorithm is correct and its running time is $\Ohtilde(\frac{\ell}{k+1}+|J|)$.
\end{restatable}
  
\begin{theorem}\label{thm:simple}
  There exists an algorithm that, given strings $X$ and $Y$, and an integer $k\ge 0$, returns
    YES if $\ED(X,Y)\le k$, and
    NO if $\ED(X,Y)> (3k+5)k$.
  With high probability,  the algorithm is correct and its running time is $\Ohtilde(\frac{1}{k+1}|X|+k^2)$.
\end{theorem}
\begin{proof}
The pseudocode is given in \cref{alg:simple}. Queries $\apLCE{k}$ are implemented using \cref{prp:aplce}.
With high probability,  all the queries are answered correctly. Conditioned on this assumption,
\cref{lem:simple_correct} yields that \cref{alg:simple} is correct with high probability.
It remains to analyze the running time. The cost of instructions other than $\apLCE{k}$ queries is $\Oh(k)$.
By \cref{prp:aplce}, the cost of computing $d_i$ is $\Ohtilde(\frac{1}{k+1}({d_i-d'_i})+k)$.
Due to $0\le d'_0\le d_0 \le d'_1\le d_1 \le \cdots \le d'_k \le d_k \le |X|$, this sums up to $\Ohtilde(\frac{1}{k+1}|X|+k^2)$ across 
all queries.
\end{proof}

\subsection{Proof of \cref{prp:aplce}}\label{sec:apLCE}
Our implementation of $\max_{j\in J}\apLCE{k}^{X,Y}(i,j)$ queries heavily borrows from~\cite{GKS19}.
However, our problem is defined in a more abstract way and we impose stricter conditions 
on the output value, so we cannot use tools from~\cite{GKS19} as black boxes; thus, we opt for a self-contained presentation.

On the highest level, in \cref{lem:pm}, we develop an oracle that, additionally given a threshold $\ell$,
must return YES if $\max_{j\in J}\LCE_0^{X,Y}(i,j) \ge \ell$,
must return NO if  $\max_{j\in J}\LCE_k^{X,Y}(i,j) < \ell$, and may return an arbitrary answer otherwise.
The final algorithm behind \cref{prp:aplce} is then an exponential search on top of the oracle.
This way, we effectively switch to the decision version of the problem, 
which is conceptually and technically easier to handle.

The oracle can be specified as follows:
it must return YES if $X[i\dd i+\ell)=Y[j\dd j+\ell)$ for some $j\in J$,
and NO if $\HAM(X[i\dd i+\ell), Y[j\dd j+\ell))>k$ for every $j\in J$.
Now, if $\ell\le 3|J|$, then we can afford running a classic exact pattern matching algorithm~\cite{MP70}
to verify the YES-condition.
Otherwise, we use the same method to filter \emph{candidate positions} $j\in J$
satisfying $X[i\dd i+3|J|)=Y[j\dd j+3|J|)$.
If there is just one candidate position $j$, we can continue checking it by comparing $X[i+s]$ and $Y[j+s]$ at shifts $s$ sampled uniformly at random with rate $\Thetatilde(\frac{1}{k+1})$.

If there are many candidate positions, \cref{fct:period} implies that $X[i\dd i+3|J|)$ is periodic with period $p\le|J|$
and that the candidate positions form an arithmetic progression with difference $p$.
We then check whether $p$ remains a period of $X[i\dd i+\ell)$, and of $Y[j\dd j+\ell)$ for the leftmost candidate~$j$.
Even if either check misses $\frac k2$ mismatches with respect to the period,
two positive answers guarantee $\HAM(X[i\dd i+\ell),\allowbreak Y[j\dd j+\ell))\le k$, which lets us return YES.
Thus, the periodicity check (\cref{lem:period}) can be implemented by testing individual positions sampled with rate $\Thetatilde(\frac{1}{k+1})$. 

A negative answer of the periodicity check is witnessed by a single mismatch with respect to the period.
However, further steps of the oracle require richer structure as a leverage. Thus,
we augment the periodicity check so that it returns a \emph{break} $B$ with $|B|=2|J|$ and $\per(B)>|J|$.
For this, we utilize a binary-search-based procedure, which is very similar to finding ``\emph{period transitions}'' in~\cite{GKS19}.
Whenever $X[i\dd i+\ell)=Y[j\dd j+\ell)$, the break $B$ (contained in either string) must match exactly
the corresponding fragment in the other string. 
Since the break is short, we can afford checking this match for every $j\in J$ (using exact pattern matching again), and since it is not periodic, at most one candidate position $j\in J$ passes this test.
This brings us back to the case with at most one candidate position. 

Compared to the outline above, the algorithm described in \cref{lem:pm} handles the two main cases (many candidate positions vs one candidate position) in a uniform way, which simplifies formal analysis and implementation details.

We start with the procedure that certifies (approximate) periodicity or finds a break.

\begin{lemma}\label{lem:period}
  There exists an algorithm that, given a string $T$ an integer $k\ge 0$,
  and a positive integer $q\le \frac12|T|$, returns either
\begin{itemize}
  \item a length-$2q$ \emph{break} $B$ in $T$ such that $\per(B)>q$, or
  \item $\bot$, certifying that $p:=\per(T[0\dd 2q))\le q$ and $|\{i\in [0\dd |T|) : T[i]\ne T[i\bmod p]\}|\le k$.
\end{itemize}
With high probability, the algorithm is correct and costs $\Ohtilde(\frac{1}{k+1}|T|+q)$ time.
\end{lemma}
\SetKwFunction{FindBreak}{FindBreak}
\begin{algorithm}[t!]

$p := \per(T[0\dd 2q))$\;\label{ln:FBper}
\lIf{$p > q$}{\Return{$T[0\dd 2q)$}}
Let $S\sub [0\dd |T|)$ with elements sampled independently at sufficiently large rate~$\Thetatilde(\frac{1}{k+1})$\;
\ForEach{$s\in S$}{
  \If{$T[s] \ne T[s\bmod p]$}{\label{ln:FBtest}
    $b:=2q$; $e := s$\;\label{ln:FBbeg}
    \While{$b < e$}{\label{ln:FBbin}
      $m := \ceil{\frac{b+e}{2}}$\;
      \For{$j := m-2q$ \KwSty{to} $m-1$}{
        \lIf{$T[j] \ne T[j\bmod p]$}{$e := j$}\label{ln:FBe}
      }
      \lIf{$e \ge m$}{$b := m$}\label{ln:FBb}
    }
    \Return{$T(b-2q\dd b]$}\label{ln:FBY}
  }
}
\Return{$\bot$}\label{ln:FBbot}
\caption{\protect\FindBreak{$T$, $q$, $k$}}\label{alg:period}
\end{algorithm}

\begin{proof}
A procedure $\FindBreak{$T$, $q$, $k$}$ implementing \cref{lem:period}
is given as \cref{alg:period}.

First, the algorithm computes the shortest period $p = \per(T[0\dd 2q))$. 
If $p > q$, then the algorithm returns $B:=T[0\dd 2q)$, which is a valid break due to $\per(T)=p>q$.

Otherwise, the algorithm tries to check if $\bot$ can be returned.
If we say that a position $i\in [0\dd |T|)$ is \emph{compatible} when $T[i] = T[i \bmod p]$,
then $\bot$ can be returned provided that there are at most $k$ incompatible positions.
The algorithm samples a subset $S\sub [0\dd |T|)$ with a sufficiently large rate
$\Ohtilde(\frac{1}{k+1})$. 
Such sampling rate guarantees that if there are at least $k+1$ incompatible positions,
then  with high probability at least one of them belongs to $S$.
Consequently, the algorithm checks whether all positions $s\in S$ are compatible (\cref{ln:FBtest}),
and, if so, returns $\bot$ (\cref{ln:FBbot}); this answer is correct with high probability.

In the remaining case, the algorithm constructs a break $B$ based on an incompatible position $s$
(Lines~\ref{ln:FBbeg}--\ref{ln:FBY}).
The algorithm performs a binary search maintaining positions $b,e$ with $2q \le b \le e < |T|$ such that $e$ is incompatible and positions in $[b-2q\dd b)$ are all compatible.
The initial choice of $b := 2q$ and $e := s$ satisfies the invariant because positions in $[0\dd 2q)$ are all 
compatible due to $p=\per(T[0\dd 2q))$.
While $b < e$, the algorithm chooses $m := \ceil{\frac{b+e}{2}}$.
If $[{m-2q}\dd\allowbreak m)$ contains an incompatible position $j$, then $j\ge b$ (because $j\ge m-2q\ge b-2q$ and positions in $[b-2q\dd b)$ are all compatible), so the algorithm maintains the invariant setting $e:=j$ for such a position $j$ (\cref{ln:FBe}). 
Otherwise, positions in $[m-2q\dd m)$ are all compatible. Due to $m\le e$,
this means that the algorithm maintains the invariant setting  $b:= m$ (\cref{ln:FBb}).
Since $e-b$ decreases by a factor of at least two in each iteration,
after $\Oh(\log |T|)$ iterations, the algorithm obtains $b = e$.
Then, the algorithm returns $B:=T(b-2q\dd b]$.

We shall prove that this is a valid break.
For a proof by contradiction, suppose that $p':= \per(B) \le q$. 
Then, $p'$ is also period of $T(b-2q\dd b)$.
Moreover, the invariant guarantees that positions in $(b-2q\dd b)$ are all compatible, 
so also $p$ is a period of $T(b-2q\dd b)$.
Since $p+p'-1 \le 2q-1$, the periodicity lemma~\cite{FW65} implies that also $\gcd(p,p')$ is a period of $X(b-2q\dd b)$.
Consequently, $T[b] = T[{b-p'}]=T[b-p]=T[(b-p)\bmod p]=T[b\bmod p]$,
i.e., $b$ is compatible. However, the invariant assures that $b$ is incompatible.
This contradiction proves that $\per(B)>q$.

It remains to analyze the running time.
Determining $\per(T[0\dd 2q))$ in \cref{ln:FBper} costs $\Oh(q)$ time using a classic algorithm~\cite{MP70}.
The number of sampled positions is $|S|=\Ohtilde(\frac{1}{k+1}|T|)$ with high probability,
so the test in \cref{ln:FBtest} costs  $\Ohtilde(\frac{1}{k+1}|T|)$ time in total.
Binary search (the loop in \cref{ln:FBbin}) has $\Oh(\log |T|)=\Ohtilde(1)$ iterations,
each implemented in $\Oh(q)$ time. 
The total running time is $\Ohtilde(\frac{1}{k+1}|T|+q)$.
\end{proof}

Next comes the oracle testing $\max_{j\in J}\apLCE{k}^{P,T}(i,j)\le \ell$. 

\begin{lemma}\label{lem:pm}
There exists an algorithm that, given strings $X$ and $Y$, an integer $k\ge 0$, an integer $\ell>0$,
an integer $i\in [0\dd |X|-\ell]$, and a non-empty range $J\sub [0\dd |Y|-\ell]$,
returns YES if $\exists_{j\in J}: X[i\dd i+\ell)=Y[j\dd j+\ell)$,
and NO if $\forall_{j\in J}: \HAM(X[i\dd i+\ell),Y[j\dd j+\ell)) > k$.
With high probability,  the algorithm is correct and its running time is $\Ohtilde(\frac{\ell}{k+1}+|J|)$.
\end{lemma}
\SetKwFunction{PM}{Oracle}
\begin{proof}
A procedure $\PM{$P$, $T$, $i$, $J$, $k$, $\ell$}$ implementing \cref{lem:pm}
is given as \cref{alg:pm}.

\paragraph*{Algorithm}
If $\ell < 3|J|$, then the algorithm simply returns the answer based on
whether $X[i\dd i+\ell)=Y[j\dd j+\ell)$  holds for some $j\in J$.
Otherwise, the algorithm computes a set $C\sub J$ of \emph{candidate positions} $j$
satisfying $X[i\dd i+3|J|)=Y[j\dd j+3|J|)$, and returns NO if $C=\emptyset$.
In the remaining case, the algorithm applies the procedure $\textsf{FindBreak}$ of \cref{lem:period} to $X[i\dd i+\ell)$ and $Y[\max C \dd \min C+\ell)$, both with $q=|J|$ and threshold $\floor{\frac k2}$. If both strings are certified to have an approximate period, then the algorithm returns YES.
Otherwise, the algorithm further filters $C$ using the breaks returned by $\textsf{FindBreak}$:
If a break $B_X=X[x \dd x')$ is found in $X[i\dd i+\ell)$,
then $C$ is restricted to positions $j$ satisfying $B_X = Y[j-i+x\dd \allowbreak j-i+x')$.
Similarly, if a break $B_Y=Y[y\dd y')$ is found in $Y[\max C \dd \min C + \ell)$,
then $C$ is restricted to positions $j$ satisfying $B_Y = X[i-j+y\dd i-j+y')$.
If this filtering leaves $C$ empty, then the algorithm returns NO.
Otherwise, the algorithm samples a subset $S\sub [0\dd \ell)$ with sufficiently large rate $\Ohtilde(\frac{1}{k+1})$,
and returns the answer depending on whether $X[i+s]=Y[\min C + s]$ holds for all $s\in S$.

\paragraph*{Correctness}
Denote $M = \{j \in J : X[i\dd i+\ell)=Y[j\dd j+\ell)\}$. Recall that the algorithm
must return YES if $M\ne \emptyset$, and it may return NO whenever $M=\emptyset$.

If $|J|<3\ell$, then the algorithm verifies $M\ne \emptyset$, so the answers are correct. 
Thus, we henceforth assume $|J|\ge 3\ell$.

Let us argue that $M \sub C \sub J$ holds throughout the execution:
indeed, every position $j\in M$ satisfies $X[i\dd \allowbreak i+3|J|)=Y[j\dd j+3|J|)$,
as well as $X[x\dd x')=Y[{j-i+x}\dd\allowbreak j-i+x')$ for every fragment $X[x\dd x')$ contained in $X[i\dd i+\ell)$,
and $Y[y\dd y')=X[i-j+y\dd\allowbreak i-j+y')$ for every fragment $Y[y\dd y')$ contained in $Y[j\dd j+\ell)$.
Moreover, the strings in the two calls to $\textsf{FindBreak}$ are chosen so that the breaks, if any,
are contained in $X[i\dd i+\ell)$, and in $Y[j\dd j+\ell)$ for every $j\in C$, respectively. 
Consequently, the NO answers returned in \cref{ln:PMNo1,ln:PMNo2} are correct.

Next, note that the calls to $\textsf{FindBreak}$ satisfy the requirements of \cref{lem:period}.
In particular, the two strings are of length at least $3|J|$ and $2|J|$, respectively.
To justify the YES answer in \cref{ln:PMYES}, we shall prove that $\HAM(X[i\dd i+\ell),Y[\min C \dd \min C + \ell))\le k$ holds with high probability in case both calls return $\bot$.
Denote $p = \per(X[i\dd i+2|J|])$, let $P=X[i\dd i+p)$ be the corresponding string period,
and let $P^\infty$ be the concatenation of infinitely many copies of $P$.
The outcome $\bot$ of the first call to $\textsf{FindBreak}$ certifies that $X[i\dd i+\ell)$
is with high probability at Hamming distance at most $\frac k2$ from a prefix of $P^\infty$.
Due to $X[i\dd i+2|J|)=Y[\max C \dd \max C + 2|J|)$, 
the outcome $\bot$ of the second call to $\textsf{FindBreak}$ certifies that also $Y[\max C \dd \min C+\ell)
$ is with high probability at Hamming distance at most $\frac k2$ from a prefix of $P^\infty$.
Moreover, by \cref{fct:period}, $p$ is a divisor of $\max C - \min C$, 
so, $Y[\min C \dd \max C)$ is an integer power of $P$.
Thus, $Y[\min C \dd \min C + \ell)$ is with high probability at Hamming distance at most $\frac k2$ from a prefix of $P^\infty$.
Now, the triangle inequality yields $\HAM(X[i\dd i+\ell),Y[\min C \dd \min C + \ell))\le k$, as claimed.

It remains to justify the answers returned in \cref{ln:PMtest,ln:PMyes}.
Because the breaks $B_X$ and $B_Y$, if defined, satisfy $\per(B_X)>|J|$ and $\per(B_Y)>|J|$, 
\cref{lem:period} implies that their exact occurrences must be more than $|J|$ positions apart.
Consequently, applying \cref{ln:PMF1} or \cref{ln:PMF2}
leaves at most one position in $C$.
Thus, the algorithm correctly returns NO if it detects a mismatch in \cref{ln:PMtest} while testing random shifts $s$
for the unique position $\min C \in C$.
Finally, note that the sampling rate in the construction of $S$ 
guarantees that if there are at least $k+1$ mismatches between $X[i\dd i+\ell)$ and $Y[\min C\dd \min C+\ell)$, then  with high probability at least one of them is detected.
Thus, returning YES in \cref{ln:PMyes} is also correct.

\begin{algorithm}[t]
  \If{$\ell <3|J|$}{\Return $\exists_{j\in J}: X[i\dd i+\ell)=Y[j \dd j+\ell)$\;}\label{ln:PMpm}
  $C := \{j \in J : X[i\dd i+3|J|) = Y[j \dd j+3|J|)\}$\;\label{ln:PMpref}
  \lIf{$C = \emptyset$}{\Return NO}\label{ln:PMNo1}
  $B_X := \FindBreak{$X[i\dd i+\ell)$, $|J|$, $\floor{\frac k2}$}$\;
  $B_Y := \FindBreak{$Y[\max C\dd \min C{+}\ell)$, $|J|$, $\floor{\frac k2}$}$\;
  \lIf{$\bot = B_X$ \KwSty{and} $\bot = B_Y$}{\Return YES}\label{ln:PMYES}
  \lIf{$\bot \ne B_X =: X[x\dd x')$}{%
    $C := \{j\in C : B_X = Y[j-i+x\dd j-i+x')\}$%
  }\label{ln:PMF1}
  \lIf{$\bot \ne B_Y =: Y[y \dd y)$}{%
    $C := \{j \in C : B_Y = X[i-j+y\dd i-j+y')\}$%
  }\label{ln:PMF2}
  \lIf{$C = \emptyset$}{\Return NO}\label{ln:PMNo2}
  Let $S\sub\![0\dd \ell)$ with elements sampled independently at sufficiently large rate $\Thetatilde(\frac{1}{k+1})$\;
  \ForEach{$s\in S$}{
    \lIf{$X[i+s] \ne Y[\min C+s]$}{\Return NO}\label{ln:PMtest}
  }
  \Return YES\;\label{ln:PMyes}
  \caption{\protect\PM{$X$, $Y$, $i$, $J$, $k$, $\ell$}}\label{alg:pm}
  \end{algorithm}

\paragraph*{Running time}
\cref{ln:PMpm,ln:PMpref,ln:PMF1,ln:PMF2} can be interpreted as finding exact occurrences of $X[i\dd i+\ell)$,
$X[i\dd i+3|J|)$, $B_X$, and $B_Y$, respectively, starting at up to $|J|$ consecutive positions of $X$ or $Y$.
Since the length of all these patterns is $\Oh(|J|)$, this search can be implemented in $\Oh(|J|)$ 
using a classic pattern matching algorithm~\cite{MP70}.
The calls to \FindBreak\ from \cref{lem:period} cost $\Ohtilde(\frac{\ell}{\floor{k/2}+1}+|J|)$ time with high probability.
Finally, the number of sampled positions is $|S|=\Ohtilde(\frac{\ell}{k+1})$ with high probability,
and this is also the total cost of \cref{ln:PMtest}.
The total running time is $\Ohtilde(\frac{\ell}{k+1}+|J|)$.
\end{proof}

Finally, we derive \cref{prp:aplce} via a simple reduction to \cref{lem:pm}.
\aplce*
\begin{proof}
  Observe that \cref{lem:pm} provides an oracle that returns YES if $\max_{j\in J}\LCE_0^{X,Y}(i,j)\ge \ell$
  and NO if $\max_{j\in J}\LCE_k^{X,Y}(i,j)< \ell$.
  However, before calling $\PM{$P$, $T$, $i$, $J$, $k$, $\ell$}$, we need to make sure that $\ell>0$, $i\in [0\dd |X|-\ell]$, and $\emptyset \ne J \sub [0\dd |Y|-\ell]$.
  Thus, basic corner cases have to be handled separately:
  The algorithm returns YES if $\ell \le 0$;
  otherwise, it sets $J:= J \cap [0\dd |Y|-\ell]$,
  returns NO if $i\notin [0\dd |X|-\ell]$
  or $J=\emptyset$, and makes a call $\PM{$P$, $T$, $i$, $J$, $k$, $\ell$}$ in the remaining case.

  A single call to the oracle costs $\Ohtilde(\frac{\ell}{k+1} + |J|)$ time.
  Hence, we need to make sure that the intermediate values of the threshold $\ell$
  are bounded from above by a constant multiple of the final value.
  For this, the algorithm uses exponential search rather than ordinary binary search.
\end{proof}

\section{Improved Approximation Ratio}\label{sec:edLCE}
Goldenberg et al.~\cite{GKS19} generalized their algorithm in order to solve
the $k$ vs $\alpha k$ gap edit distance problem in $\Ohtilde(\frac{n}{\alpha}+k^3)$
time for any $\alpha\ge 1$.
This transformation is quite simple, because \cref{alg:lv} (the Landau--Vishkin algorithm)
with $\LCE_0$ queries replaced by $\apLCE{\alpha-1}$ queries 
returns YES if $\ED(X,Y)\le k$ and NO if $\ED(X,Y)> k+(\alpha-1)(k+1)$.

However, if we replace $\apLCE{k}$ queries with $\apLCE{\alpha-1}$ queries in \cref{alg:simple},
then we are guaranteed to get a NO answer only if $\ED(X,Y)>2k(k+2)+(\alpha-1)(k+1)$.
As a result, with an appropriate adaptation of \cref{prp:aplce},
\cref{alg:simple} yields an $\Ohtilde(\frac{n}{\alpha}+k^2)$-time solution to the $k$ vs $\alpha k$ gap edit distance problem only for $\alpha =\Omega(k)$. 
The issue is that \cref{alg:simple} incurs a cost of up to $\Theta(k)$ edits for up to $\Theta(k)$ arbitrary
changes of the shift $y-x$ within queries $\LCE^{X,Y}(x,y)$. 
On the other hand,  no such shift changes are performed in \cref{alg:lv}, but this results 
in $\LCE^{X,Y}(x,y)$ queries asked for up to $\Theta(k^2)$ distinct positions $x$, which is the reason behind the $\Ohtilde(k^3)$ term in the running time $\Ohtilde(\frac{n}{\alpha}+k^3)$ of~\cite{GKS19}.

Nevertheless, since each $\apLCE{\alpha-1}^{X,Y}(x,y)$ query incurs a cost of up to $\alpha-1$ edits (mismatches)
it is still fine to pay $\Oh(\alpha-1)$ further edits (insertions or deletions) to change the shift $y-x$
by up to $\alpha-1$. Hence, we design \cref{alg:improved} as a hybrid of \cref{alg:lv,alg:simple}.

\begin{algorithm}[h]
  \lForEach{$i \in [0\dd k]$ \KwSty{and} $j \in [\floor{\frac{-k}{\alpha}}-1\dd \floor{\frac{k}{\alpha}}+1]$}{$d'_{i,j}:=d_{i,j}:=-\infty$}
  $d'_{0,0}:= 0$\;
  \For{$i := 0$ \KwSty{to} $k$}{
    \For{$j := \floor{\frac{-k}{\alpha}}$ \KwSty{to} $\floor{\frac{k}{\alpha}}$}{
      \If{$d'_{i,j}\ne -\infty$}{%
        $d_{i,j} := d'_{i,j}+\max_{\delta=j\alpha}^{(j+1)\alpha-1}\apLCE{\alpha-1}^{X,Y}(d'_{i,j},d'_{i,j}+\delta)$\;
      }%
    }  
    \For{$j := \floor{\frac{-k}{\alpha}}$ \KwSty{to} $\floor{\frac{k}{\alpha}}$}{%
      $d'_{i+1,j} := \min(|X|,\max(d_{i,j-1},d_{i,j}+1,d_{i,j+1}+1))$\;
    }
  
  }
  $j := \floor{\frac{1}{\alpha}(|Y|-|X|)}$\;
  \lIf{$||X|-|Y|| \le k$ \KwSty{and} $d_{k,j} = |X|$}{\Return{YES}}
  \lElse{\Return{NO}}
  \caption{Improved algorithm}\label{alg:improved}
  \end{algorithm}

\begin{lemma}\label{lem:improved_correct}
  For any integers $k\ge 0$ and $\alpha \ge 1$, \cref{alg:improved} returns YES if $\ED(X,Y)\le k$ and NO if $\ED(X,Y)>k+3(k+1)(\alpha-1)$.
\end{lemma}
\begin{proof}
  As in the proof of \cref{lem:simple_correct}, we characterize the values $d_{i,j}$ and $d'_{i,j}$ using two claims.

\begin{claim}\label{clm:improved_no}
  Each $i\in [0\dd k]$ and  $j \in [\floor{\frac{-k}{\alpha}}\dd \floor{\frac{k}{\alpha}}]$
  satisfies the following two properties:
  \begin{enumerate}[label={\rm(\alph*)}]
    \item\label{it:improved_dij} $\ED(X[0\dd d'_{i,j}),Y[0\dd y)) \le i+(3i+1)(\alpha-1)$ for $y\in [d'_{i,j}+j\alpha\dd d'_{i,j}+(j+1)\alpha)\cap [0\dd |Y|]$;
    \item\label{it:improved_dijp} $\ED(X[0\dd d_{i,j}),Y[0\dd y)) \le i+3(i+1)(\alpha-1)$ for $y\in [d_{i,j}+j\alpha\dd d_{i,j}+(j+1)\alpha)\cap [0\dd |Y|]$.
  \end{enumerate}
\end{claim}
\begin{proof}
We proceed by induction on $i$. Our base case is Property~\ref{it:improved_dij} for $i=0$.
Due to $d'_{0,j}=-\infty$ for $j\ne 0$, the range for $y$ is non-empty only for $j=0$,
when the range is $[0\dd \alpha)$ due to $d'_{0,0}=0$.
Moreover, for $y\in [0\dd \alpha)$, we have $\ED(X[0\dd d'_{0,0}),Y[0\dd y))=y \le \alpha-1$.

Next, we shall prove Property~\ref{it:improved_dijp} for $i\ge 0$ assuming that Property~\ref{it:improved_dij}
is true for $i$.
By definition of $\apLCE{\alpha-1}$ queries, we have $d_{i,j} \le d'_{i,j}+\LCE_{\alpha-1}^{X,Y}(d'_{i,j},y')$ 
for some position $y'\in [d'_{i,j}+j\alpha\dd d'_{i,j}+(j+1)\alpha)\cap [0\dd |Y|]$,
and thus $\HAM(X[d'_{i,j}\dd d_{i,j}),Y[y'\dd y'+d_{i,j}-d'_{i,j}))\le \alpha-1$.
The inductive assumption yields
$\ED(X[0\dd d'_{i,j}),Y[0\dd y'))\le i+(3i+1)(\alpha-1)$,
so we have $\ED(X[0\dd d_{i,j}),Y[0\dd y'+d_{i,j}-d'_{i,j}))\le i+(3i+2)(\alpha-1)$.
Due to $|y'+d_{i,j}-d'_{i,j}-y|\le \alpha-1$, we conclude that 
$\ED(X[0\dd d_{i,j}),Y[0\dd y))\le i+3(i+1)(\alpha-1)$.

Finally, we shall prove Property~\ref{it:improved_dij} for $i>0$ assuming that Property~\ref{it:improved_dijp}
is true for $i-1$. We consider three subcases:
If $d'_{i,j} \le d_{i-1,j-1}$, then the inductive assumption yields 
$\ED(X[0\dd d'_{i,j}),\allowbreak Y[0\dd y-\alpha))\le(i-1)+3i(\alpha-1)$,
and therefore 
$\ED(X[0\dd d'_{i,j}),Y[0\dd y))\le \alpha+(i-1)+3i(\alpha-1) = i+(3i+1)(\alpha-1)$.
If $d'_{i,j} \le d_{i-1,j}+1$, then the inductive assumption yields 
$\ED(X[0\dd d'_{i,j}-1),Y[0\dd y-1))\le {(i-1)}+3i(\alpha-1)$,
and therefore 
$\ED(X[0\dd d'_{i,j}),Y[0\dd y))\le 1+ (i-1)+3i(\alpha-1) = i+3i(\alpha-1)$.
If $d'_{i,j} \le d_{i-1,j+1}+1$, then the inductive assumption yields 
$\ED(X[0\dd d'_{i,j}-\alpha),Y[0\dd y))\le (i-1)+3i(\alpha-1)$,
and therefore 
$\ED(X[0\dd d'_{i,j}),Y[0\dd y))\le \alpha+(i-1)+3i(\alpha-1) = i+(3i+1)(\alpha-1)$.
\end{proof}
In particular, if the algorithm returns YES, then $\ED(X,Y)\le k + 3(k+1)(\alpha-1)$.

\begin{claim}\label{clm:improved_yes}
  If $\ED(X[0\dd x),Y[0\dd y)) = i$ for $x\in [0\dd |X|]$, $y\in [0\dd |Y|]$, and $i\in [0\dd k]$,
  then $x \le d_{i,j}$ holds for $j=\floor{\frac{1}{\alpha}(y-x)}$.
  \end{claim}
  \begin{proof}
  We proceed by induction on $i$. Both in the base case of $i=0$ and in the inductive step of $i>0$,
  we prove that $x \le d'_{i,j}+\max_{\delta=j\alpha}^{(j+1)\alpha-1}\LCE_0^{X,Y}(d'_{i,j}, d'_{i,j}+\delta)$.
  This implies the claim since $d_{i,j} \ge d'_{i,j}+\max_{\delta=j\alpha}^{(j+1)\alpha-1}\LCE_0^{X,Y}(d'_{i,j}, d'_{i,j}+\delta)$,
  holds by definition of $\apLCE{\alpha-1}$ queries.
  
  In the base case of $i=0$, we have $X[0\dd x)=Y[0\dd y)$, so $x=y$ and  $j=0$.
  Consequently, due to $d'_{0,0}=0$,  we have 
  $x\le \LCE_0^{X,Y}(0,0) \le d'_{0,0}+\max_{\delta=0}^{\alpha-1}\LCE_0^{X,Y}(d'_{0,0},d'_{0,0}+\delta)$.
  
  For $i > 0$, we consider an optimal alignment between $X[0\dd x)$ and $Y[0\dd y)$, and we distinguish its maximum prefix with $i-1$ edits.
  This yields positions $x',x''\in [0\dd x]$ and $y',y''\in [0\dd y]$ with $x''-x'\in \{0,1\}$ and $y''-y'\in \{0,1\}$ such that $\ED(X[0\dd x'),Y[0\dd y'))=i-1$ and $X[x''\dd x)=Y[y''\dd y)$.
  The inductive assumption yields $x' \le d_{i-1,j'}$, where $j'=\floor{\frac{1}{\alpha}(y'-x')}$
  satisfies $|j-j'|\le 1$.
  We shall prove that $x'' \le d'_{i,j}$ by considering two possibilities. 
  If $j' \ge j$, then $x'' \le \min(x,x'+1) \le \min(|X|,d_{i-1,j'}+1)\le d'_{i,j}$.
  If $j' < j$, on the other hand, then $y'-x' < y''-x''$ implies 
  $x'' = x' \le d_{i-1,j'}=d_{i-1,j-1} \le d'_{i,j}$.
  Due to $X[x''\dd x)=Y[y''\dd y)$, we have $\LCE_0^{X,Y}(x'',y'')\ge x-x''$.
  By $x'' \le d'_{i,j}$, this implies $\LCE_0^{X,Y}(d'_{i,j}, d'_{i,j}+y-x)\ge x-d'_{i,j}$. 
  By definition of $j$, we conclude that
  $x = d'_{i,j}+(x-d'_{i,j}) \le d'_{i,j} + \LCE_0^{X,Y}(d'_{i,j}, d'_{i,j}+y-x) \le d'_{i,j}+\max_{\delta=j\alpha}^{(j+1)\alpha-1}\LCE_0^{X,Y}(d'_{i,j}, d'_{i,j}+\delta)$.
  \end{proof}
  In particular, if $\ED(X,Y)\le k$, then the algorithm returns YES.
\end{proof}

If we use \cref{prp:aplce} to implement $\apLCE{\alpha-1}$ queries in \cref{alg:improved},
then the cost of computing $d_{i,j}$ is $\Oh(\alpha + \frac{1}{\alpha}(d_{i,j}-d'_{i,j}))$
with high probability. This query is performed only for $d'_{i,j}\ge 0$,
and it results in $d_{i,j}\le |X|$. As $d_{i,j}\le d'_{i+1,j}$, the total query time
for fixed $j$ sums up to $\Ohtilde(\frac{1}{\alpha}|X|+k\alpha)$ across 
all queries.
Over all the $\Oh(\frac{k}{\alpha})$ values $j$, this gives  $\Ohtilde(\frac{k}{\alpha^2}|X|+k^2)$
time with high probability,
which is not comparable to 
the running time $\Ohtilde(\frac{1}{\alpha}|X|+k^3)$ of~\cite{GKS19}.

However, we can obtain a faster algorithm using the data structure specified below and described in \cref{sec:batch}.
In particular, this result dominates \cref{prp:aplce} and,
if we set $\Delta=[-k\dd k]$, then $\apLCE{k}^{X,Y}(x,y)$ queries with $|x-y|\le k$
can be answered in $\Ohtilde(k)$ time after $\Ohtilde(\frac{1}{k+1}|X|)$ preprocessing,
as promised in \cref{sec:simple}.

\begin{restatable}{proposition}{prpbatch}\label{prp:batch}
  There exists a data structure that, initialized with strings $X$ and $Y$, an integer $k\ge 0$,
  and an integer range $\Delta$,
  answers the following queries: given an integer $x$, return $\apLCE{k}^{X,Y}(x,x+\delta)$
  for all $\delta \in \Delta$.
  The initialization costs $\Ohtilde(\frac{1}{k+1}|X|)$ time with high probability,
  and the queries cost $\Ohtilde(|\Delta|)$ time with high probability.
\end{restatable}

Since the $\apLCE{\alpha-1}^{X,Y}(x,y)$ queries in \cref{alg:improved} are asked for $\Oh(\frac{k^2}{\alpha})$ positions $x$ and for positions $y$ satisfying $|y-x|=\Oh(k)$, a straightforward application of \cref{prp:batch}
yields an $\Ohtilde(\frac{1}{\alpha}|X|+\frac{k^3}{\alpha})$-time implementation of \cref{alg:improved},
which is already better the running time of~\cite{GKS19}.
However, the running time of a more subtle solution described below subsumes both $\Ohtilde(\frac{1}{\alpha}|X|+\frac{k^3}{\alpha})$ and $\Ohtilde(\frac{k}{\alpha^2}|X|+k^2)$ (obtained using \cref{prp:aplce}).

\begin{theorem}\label{thm:tradeoff}
  There exists an algorithm that, given strings $X$ and $Y$, an integer $k\ge 0$,
  and a positive integer $\alpha = \Oh(k)$, returns
      YES if $\ED(X,Y)\le k$, and
       NO if $\ED(X,Y)> k+3(k+1)(\alpha-1)$.
  With high probability,  the algorithm is correct and its running time is $\Ohtilde(\frac{1}{\alpha}|X|+k^2+\frac{k}{\alpha}\sqrt{|X|k})$.
\end{theorem}
\begin{proof}
  We define an integer parameter $b\in [1\dd \lceil\frac{k}{\alpha}\rceil]$ (to be fixed later) and initialize $\Oh(\frac{k}{\alpha b})$ instances of the data structure of \cref{prp:batch} for answering $\apLCE{\alpha-1}^{X,Y}$ queries. The instances are indexed with $j'\in  [\floor{\frac{-k}{\alpha b}}\dd \floor{\frac{k}{\alpha b}}]$, and the $j'$th instance has interval $\Delta_{j'}=[j'\alpha b\dd (j'+1)\alpha b)$.
  This way, the value  $d_{i,j}$  can be retrieved from the values $\apLCE{\alpha-1}(d'_{i,j},d'_{i,j}+\delta)$ for $\delta\in \Delta_{\lfloor\frac{j}{b}\rfloor}$, that is, from a single query to an instance of the data structure of \cref{prp:batch}.

  Correctness follows from \cref{lem:improved_correct} since with high probability all $\apLCE{\alpha-1}$ queries
  are answered correctly.  
  The total preprocessing cost is $\Ohtilde(\frac{k}{\alpha b}\cdot \frac{1}{\alpha }|X|)=\Ohtilde(\frac{k}{\alpha ^2b}|X|)$ with high probability,
  and each value $d_{i,j}$ is computed in $\Ohtilde(\alpha b)$ time with high probability.
  The number of queries is $\Oh(\frac{k^2}{\alpha})$, so the total running time is $\Ohtilde(\frac{k}{\alpha^2b}|X|+k^2 b)$  with high probability.
  Optimizing for $b$ yields $\Ohtilde(\frac{k}{\alpha}\sqrt{|X|k})$. Due to $b\in [1\dd \lceil\frac{k}{\alpha}\rceil]$,
  we get additional terms $\Ohtilde(k^2+\frac{1}{\alpha}|X|)$.
\end{proof}

\subsection{Proof of \cref{prp:batch}}\label{sec:batch}
While there are many similarities between the proofs of \cref{prp:aplce,prp:batch},
the main difference is that we heavily rely on $\bpLCE_r$ queries in the proof of \cref{prp:batch}. 
The following fact illustrates their main advantage compared to $\apLCE{k}$ queries: \emph{composability}.
\begin{fact}\label{lem:combine}
Let $X,X',Y$ be strings, let $r > 0$ be real parameter,
and let $j\in [0\dd |Y|-|X|]$.
Suppose that $\bpLCE_r^{X,Y}(0,j)$ and $\bpLCE_r^{X',Y}(0,j+|X|)$ are independent random variables,
and define
\[ \ell:= \begin{cases} 
  \bpLCE_r^{X,Y}(0,j) & \text{if }\bpLCE_r^{X,Y}(0,j)<|X|,\\
  |X|+\bpLCE_r^{X',Y}(0,j+|X|) & \text{otherwise.}
\end{cases}
\]
Then, $\ell$ satisfies the conditions for $\bpLCE_r^{XX',Y}(0,j)$.
\end{fact}
\begin{proof}
Define $d = \HAM(X, Y[j\dd j+|X|))$ and note that the following equality holds for $k\ge 0$:
\[LCE_k^{XX',Y}(0,j) =
\begin{cases}
  \LCE_k^{X,Y}(0,j) & \text{if }k < d,\\
  |X|+\LCE_{k-d}^{X',Y}(0,j+|X|) & \text{if }k \ge d.
\end{cases}
\]

Let us first prove that $\ell \ge \LCE_0^{XX',Y}(0,j)$.
If $\bpLCE_r^{X,Y}(0,j)<|X|$, then $\LCE_0^{X,Y}(0,j) \le \bpLCE_r^{X,Y}(0,j) < |X|$
implies $d>0$, and thus $\ell = \bpLCE_r^{X,Y}\!(0,j)\ge \LCE_0^{X,Y}\!(0,j) = \LCE_0^{XX',Y}\!(0,j)$.
Otherwise, 
$\ell = |X|+\bpLCE_{r}^{X',Y}(0,j+|X|) \ge |X|+\LCE_{0}^{X',Y}(0,j+|X|)\ge \LCE_0^{XX',Y}(0,j)$.
Hence, the claim holds in both cases.

Next, let us bound the probability $\Pr\big[\ell >  \LCE_k^{XX',Y}(0,j)\big]$ for $k\ge 0$.
We consider two cases.
If $k < d$, then $\LCE_k^{XX',Y}(0,j)<|X|$ and
\begin{align*}
  \Pr\big[\ell >  \LCE_k^{XX',Y}(0,j)\big] 
  & \le \Pr\big[\bpLCE_r^{X,Y}(0,j) >  \LCE_k^{XX',Y}(0,j)\big]\\
  & = \Pr\big[\bpLCE_r^{X,Y}(0,j) >  \LCE_k^{X,Y}(0,j)\big]\\
  & \le \exp(-\tfrac{k+1}{r}).
\end{align*}
On the other hand, if $k\ge d$, then  $\LCE_k^{XX',Y}(0,j)=|X|+\LCE_{k-d}^{X',Y}(0,j+|X|)\ge |X| > \LCE_{d-1}^{X,Y}(0,j)$.
Hence, the independence of $\bpLCE_r^{X,Y}(0,j)$ and $\bpLCE_r^{X,Y}(0,j+|X|)$ yields
\begin{align*}\Pr\big[\ell >  \LCE_k^{XX',Y}\!\!(0,j)\big]
  & \le \Pr\big[\bpLCE_r^{X,Y}\!\!(0,j) >  |X| \text{ and } \bpLCE_r^{X',Y}\!\!(0,j+|X|) > \LCE_{k-d}^{X',Y}\!(0,j+|X|) \big] \\
  & = \Pr\big[\bpLCE_r^{X,Y}\!\!(0,j) >  |X|\big] \cdot \Pr\big[\bpLCE_r^{X',Y}\!\!(0,j+|X|) > \LCE_{k-d}^{X',Y}\!(0,j+|X|) \big] \\
  & \le \exp(-\tfrac{d}{r})\cdot \exp(-\tfrac{k-d+1}{r})\\
  & =\exp(-\tfrac{k+1}{r}).
\end{align*}
This completes the proof.
\end{proof}

Next, we show that a single value $\bpLCE_r^{X,Y}(0,j)$ can be computed efficiently.
We also require that the resulting position $\ell$ witnesses $\LCE_0^{X,Y}(0,j)\le \ell$.
\begin{fact}\label{fct:bplce}
There is an algorithm that, given strings $X$ and $Y$, a real parameter $r>0$,
and an integer $j$, returns a value $\ell = \bpLCE_r^{X,Y}(0,j)$ such that $X[\ell] \ne Y[j+\ell]$ or $\ell = \min(|X|,|Y|-j)$.
The algorithm takes $\Ohtilde(\frac{1}{r}|X|)$ time
with high probability.
\end{fact}
\begin{proof}
If $r\le 1$, then the algorithm returns $\ell = \LCE_0^{X,Y}(0,j)$ computed naively in $\Oh(|X|)$ time.
It is easy to see that this value satisfies the required conditions.

If $r > 1$, then the algorithm samples a subset $S \sub [0\dd \min(|X|,|Y|-j))$ so that the events
$s\in S$ are independent with $\Pr[s\in S]=\frac{1}{r}$.
If $X[s] = Y[j+s]$ for each $s\in S$, then the algorithm returns $\ell = \min(|X|,|Y|-j)$.
Otherwise, the algorithm returns $\ell = \min\{s \in S : X[s]\ne Y[j+s]\}$,
This way, $\ell \ge \LCE_0^{X,Y}(i,j)$, and  $P[\ell] \ne Y[j+\ell]$ or $\ell = \min(|X|,|Y|-j)$.

It remains to bound $\Pr\big[\ell>\LCE_k^{X,Y}(0,j)\big]$ for every $k\ge 0$. 
This event holds only if each of the $k+1$ leftmost mismatches (that is, the leftmost positions $s$ such that $X[s]\ne Y[j+s]$) does not belong to $S$. By definition of $S$, the probability of this event is 
$(1-\tfrac{1}{r})^{k+1} \le \exp(-\frac{k+1}{r})$.

Since $|S|=\Ohtilde(\frac{1}{r}\min(|X|,|Y|-j))$ with high probability, the total running time is $\Ohtilde(\frac{1}{r}|X|)$ with high probability.
\end{proof}

We are now ready to describe a counterpart of \cref{lem:period}.
\begin{lemma}\label{lem:period2}
  There is an algorithm that, given a string $T$, a real parameter $r>0$,
  and a positive integer $q \le \frac12|T|$ such that $p:=\per(T[0\dd 2q))\le q$,
  returns $\ell\in [2q\dd |T|]$ such that
\begin{itemize}
  \item $\ell = \bpLCE_r^{T,T'}(0,0)$, where $T'$ is an infinite string with $T'[i]=T[i\bmod p]$ for $i\ge 0$, and
  \item $\ell = |T|$ or $\per(T(\ell-2q\dd \ell]) > q$. 
\end{itemize}
The algorithm takes $\Ohtilde(\frac{1}{r}|T|+q)$ time with high probability.
\end{lemma}
\SetKwFunction{FindBreakTwo}{FindBreak2}
\begin{algorithm}[th]
  \SetKwComment{Comment}{$\triangleright$\ \rm}{}
  $p := \per(T[0\dd 2q))$\;\label{ln:FB2per}
  Define $T'[0\dd \infty)$ with $T'[i]=T[i\bmod p]$\;
  $\ell' := \bpLCE_r^{T,T'}(0,0)$\Comment*{computed using \cref{fct:bplce}}
  \lIf{$\ell' = |T|$}{\Return{$|T|$}}
      $b:=2q$; $e := \ell'$\;\label{ln:FB2beg}
      \While{$b < e$}{\label{ln:FB2bin}
        $m := \ceil{\frac{b+e}{2}}$\;
        \For{$j := m-2q$ \KwSty{to} $m-1$}{
          \lIf{$T[j] \ne T'[j]$}{$e := j$}\label{ln:FB2e}
        }
        \lIf{$e \ge m$}{$b := m$}\label{ln:FB2b}
      }
      \Return{$b$}\label{ln:FB2Y}\;
  \caption{\protect\FindBreakTwo{$T$, $r$, $q$}}\label{alg:period2}
  \end{algorithm}
\begin{proof}
A procedure $\FindBreakTwo{$T$, $r$, $q$}$ implementing \cref{lem:period2}
  is given as \cref{alg:period2}.

  First, the algorithm computes the shortest period $p = \per(T[0\dd 2q))$
  (guaranteed to be at most $q$ by the assumption)
  and constructs an infinite string $T'$ with $T'[i]=T[i\bmod p]$ for each 
  $i\ge 0$; note that random access to $T'$ can be easily implemented 
  on top of random access to $T$.
  
  Next, the algorithm computes $\ell':=\bpLCE_r^{T,T'}(0,0)$ using \cref{fct:bplce}.
  If $\ell' = |T|$, then $|T|$ satisfies both requirements for the resulting value $\ell$,
  so the algorithm returns $\ell:=|T|$ (\cref{ln:FBper}).

  Otherwise, the algorithm tries to find a position $\ell\le \ell'$
  such that $\per(T(\ell-2q\dd \ell])>q$ (Lines~\ref{ln:FB2beg}--\ref{ln:FB2Y}).
  This step is implemented as in the proof of \cref{lem:period}.
  We call a position $i\in [0\dd |T|)$ \emph{compatible} if $T[i]=T'[i]$.
  The algorithm performs a binary search maintaining positions $b,e$, with $2q \le b \le e < |T|$, such that $e$ is incompatible and the positions in $[b-2q\dd b)$ are all compatible.
  The initial choice of $b := 2q$ and $e := \ell$ satisfies the invariant because the positions in $[0\dd 2q)$ are all compatible due to $p=\per(T[0\dd 2q))$.
  While $b < e$, the algorithm chooses $m := \ceil{\frac{b+e}{2}}$.
  If $[m-2q\dd m)$ contains an incompatible position $j$, then $j\ge b$ (because $j\ge m-2q\ge b-2q$ and the positions in $[b-2q\dd b)$ are all compatible), so the algorithm maintains the invariant setting $e:=j$ for such a position $j$ (\cref{ln:FB2e}). 
  Otherwise, all the positions in $[m-2q\dd m)$ are compatible. Due to $m\le e$,
  this means that the algorithm maintains the invariant setting  $b:= m$ (\cref{ln:FB2b}).
  Since $e-b$ decreases at least twofold in each iteration,
  after $\Oh(\log |T|)$ iterations, the algorithm obtains $b = e$.
  Then, the algorithm returns $\ell:=b$.
  
  We shall prove that this result is correct.
  For a proof by contradiction, suppose that $p':= \per(T(b-2q\dd b]) \le q$. 
  Then, $p'$ is also period of $T(b-2q\dd b)$.
  Moreover, the invariant guarantees that positions in $(b-2q\dd b)$ are all compatible, 
  so also $p$ is a period of $T(b-2q\dd b)$.
  Since $p+p'-1 \le 2q-1$, the periodicity lemma~\cite{FW65} implies that also $\gcd(p,p')$ is a period of $T[b-2q+1\dd b)$.
  Consequently, $T[b] = T[b-p']=T[b-p]=T'[b-p]=T'[b]$,
  i.e., $b$ is compatible. However, the invariant assures that $b$ is incompatible.
  This contradiction proves that $\per(T(b-2q\dd b]) > q$.
  The incompatibility of $b$ guarantees that $\LCE_0^{T,T'}(0,0)\le b$.
  Moreover, since $b\le \ell'$, we have 
  $\Pr[b>\LCE_{k}^{T,T'}(0,0)]\le \Pr[\ell'>\LCE_{k}^{T,T'}(0,0)]\le \exp(-\frac{k+1}{r})$
  for each $k\ge 0$. Thus, $b$ satisfies the requirements for $\bpLCE_{r}^{T,T'}(0,0)$.
  
  It remains to analyze the running time.
  Determining $\per(T[0\dd 2q))$ in \cref{ln:FB2per} costs $\Oh(q)$ time using a classic algorithm~\cite{MP70}.
  The application of \cref{fct:bplce} costs $\Ohtilde(\frac{1}{r}|T|)$ time with high probability.
  Binary search (the loop in \cref{ln:FB2bin}) has $\Oh(\log |T|)=\Ohtilde(1)$ iterations,
  each implemented in $\Oh(q)$ time. 
  Consequently, the total running time is $\Ohtilde(\frac{1}{r}|T|+q)$ with high probability.
\end{proof}

Next, we develop a counterpart of \cref{lem:pm}.
\begin{lemma}\label{lem:box}
There is an algorithm that, given strings $X$ and $Y$, a real parameter $r>0$,
and an integer range $J$, 
returns $\bpLCE_r^{X,Y}(0,j)$ for each $j\in J$.
The algorithm costs $\Ohtilde(\frac{1}{r}|X|+|J|)$ time with high probability.
\end{lemma}

\SetKwFunction{PMTwo}{Batch}
\begin{algorithm}[th]
\SetKwComment{Comment}{$\triangleright$\ \rm}{}
$\Delta := \max J - \min J$\;
\lForEach{$j\in J$}{$\ell_j := \min(\LCE_0^{X,Y}(0,j),2\Delta)$}\label{ln:short}
$C := \{j\in J : \ell_j = 2\Delta \}$\;
\If{$|C|\le 1$}{
  \lForEach(\Comment*[f]{computed using \cref{fct:bplce}}){$j\in C$}{$\ell_j := \bpLCE_r^{X,Y}(0,j)$}\label{ln:only}
}\Else{\label{ln:else_start}
  $\ell^X := \FindBreakTwo{$X$, $r$, $2\Delta$}$\;
  $\ell^Y := \FindBreakTwo{$Y[\min C \dd \min(\max C + |X|,|Y|))$, $r$, $2\Delta$}$\;
  \ForEach{$j\in C$}{
    $\ell_j := \min(\ell^X, \ell^Y-j+\min C)$\;\label{ln:standard}
    \If{$\ell_j<\min(|X|,|Y|-j)$ \KwSty{and} $X(\ell_j-2\Delta\dd \ell_j]=Y(j+\ell_j-2\Delta\dd j+\ell_j]$}{\label{ln:if2}
      $\ell_j :=  \bpLCE_r^{X,Y}(0,j)$\Comment*{computed using \cref{fct:bplce}}\label{ln:special}
    }
  }\label{ln:else_end}
}
\Return{$(\ell_j)_{j\in J}$}
\caption{\protect\PMTwo{$X$, $Y$, $r$, $J$}}\label{alg:box}
\end{algorithm}
\begin{proof}
  A procedure $\PMTwo{$X$, $Y$, $r$, $J$}$ implementing \cref{lem:box}
  is given as \cref{alg:box}.

  First, the algorithm sets $\Delta := \max J - \min J$
  and computes  $\min(\LCE_0^{X,Y}(0,j),2\Delta)$ for each $j\in J$.
  Implementation details of this step are discussed later on. 
  Then, the algorithm sets $C := \{j\in J : \ell_j = 2\Delta\}$ so that 
  $\ell_j = \LCE_0^{X,Y}(0,j)$ holds for each $j\in J\sm C$.
  Consequently, $\ell_j$ can be returned as $\bpLCE_r^{P,T}(0,j)$ for $j\in J\sm C$,
  and the algorithm indeed returns these values (the values $\ell_j$ set in \cref{ln:short}
  are later modified only for $j\in C$).

  Thus, the remaining focus is on determining $\bpLCE_r^{X,Y}(0,j)$ for $j\in C$.
  If $|C|\le 1$, then these values are computed explicitly using \cref{fct:bplce}. 
  In the case of $|C|\ge 2$, handled in Lines~\ref{ln:else_start}--\ref{ln:else_end},
  the algorithm applies the \FindBreakTwo\ function of \cref{lem:period2}
  for $X$ and $\bar{Y}:=Y[\min C \dd \min(\max C + |X|,|Y|))$, both with $q=2\Delta$.

  These calls are only valid if $\per(X[0\dd 2\Delta))\le \Delta$ and $\per(\bar{Y}[0\dd 2\Delta))\le \Delta$,
  so we shall prove that these conditions are indeed satisfied.
  First, note that $C = \{o+\min J: o \in \Occ(X[0\dd 2\Delta),Y[\min J\dd \max J+2\Delta))\}$.
  Consequently, \cref{fct:period} implies that $C$ is an arithmetic progression
  with difference $p:=\per(X[0\dd 2\Delta))$.
  Due to $|C|\ge 2$, we conclude that $p\le \max J - \min J = \Delta$.
  Moreover, since $Y[j\dd j+2\Delta)=X[0\dd 2\Delta)$ for each $j\in C$,
  we have $\bar{Y}[0\dd 2\Delta)=Y[\min C\dd \min C+2\Delta)=X[0\dd 2\Delta)$,
  and thus $\per(\bar{Y}[0\dd 2\Delta))\le \Delta$.
  Hence, the calls to \FindBreakTwo\ are indeed valid.

  Based on the values $\ell^X$ and $\ell^Y$ returned by these two calls,
  the algorithm seets $\ell_j:=\min(\ell^X, \ell^Y-j+\min C)$ for each $j\in C$ in \cref{ln:standard}.
  These values satisfy the following property:
  \begin{claim}\label{clm:standard}
    For each $j\in C$, the value $\ell_j=\min(\ell^X, \ell^Y-j+\min C)$ set in \cref{ln:standard} satisfies   $\Pr[\ell_j > \LCE_k^{X,Y}(0,j)] \le  \exp(-\frac{k+1}{r})$ for every integer $k\ge 0$.
  \end{claim}
  \begin{proof}
    Note that $\ell_j \le \min(|X|,|Y|-j)$ due to $\ell^X\le |X|$ and $\ell^Y\le|\bar{Y}|$,
   Consequently, if $\LCE_k^{X,Y}(0,j)=\min(|X|,|Y|-j)$, then the claim holds trivially.
  In the following, we assume that $d := \LCE_k^{X,Y}(0,j) < \min(|X|,|Y|-j)$
  so that $X[0\dd d]$ and $Y[j\dd j+d]$ are well-defined fragments
   with $\HAM(X[0\dd d], Y[j\dd j+d])=k+1$.

   Consider an infinite string $X'$ with $X'[i]=X[i\bmod p]$ for each $i\ge 0$,
   and define $k_1 = \HAM(X[0\dd d],X'[0\dd d])$ as well as $k_2 = \HAM(Y[j\dd j+d],X'[0\dd d])$,
   observing that the triangle inequality yields $k_1 + k_2 \ge k+1$.
   Due to $\ell^X=\bpLCE_r^{X,X'}(0,0)$ (by \cref{lem:period2}), we have $\Pr[\ell^X > d] \le \exp(-\frac{k_1}{r})$, because $d \ge \LCE_{k_1-1}^{X,X'}(0,0)$.

   Next, consider an infinite string $\bar{Y}'$ with $\bar{Y}'[i] =\bar{Y}[i \bmod p]$ for $i\ge 0$,
   and observe that $\bar{Y}'=X'$ due to $\bar{Y}[0\dd 2\Delta)=X[0\dd 2\Delta)$. 
   Consequently,
   $k_2 =\HAM(\bar{Y}[j-\min C\dd j-\min C+d],\allowbreak \bar{Y}'[0\dd d])$.
   Since $C$ forms an arithmetic progression with difference $p$ and $j\in C$,
   we further have 
   $k_2 =\HAM(\bar{Y}[j-\min C\dd j-\min C+d],\bar{Y}'[j-\min C\dd j-\min C+d])$.
   Moreover, $p=\per(\bar{Y}[0\dd 2\Delta))$ and $j-\min C \le \Delta$ yields
   $\bar{Y}[0\dd j-\min C)=\bar{Y}'[0\dd j-\min C)$, and therefore
   $k_2 =\HAM(\bar{Y}[0\dd j-\min C+d],\allowbreak \bar{Y}'[0\dd j-\min C+d])$.
   We conclude that $j-\min C+d \ge \LCE_{k_2-1}^{\bar{Y},\bar{Y}'}(0,0)$.
   Due to $\ell^Y=\bpLCE_r^{\bar{Y},\bar{Y}'}(0,0)$ (by \cref{lem:period2}), we thus have $\Pr[\ell^Y-j+\min C > d]=\Pr[\ell^Y > j-\min C + d] \le \exp(-\frac{k_2}{r})$.

  Finally, since the calls to \cref{lem:period2} use independent randomness
  (and thus $\ell^X$ and $\ell^Y$ are independent random variables), we conclude that
  \begin{multline*}
    \Pr[\ell_j > \LCE_k^{X,Y}(0,j)]=
  \Pr[\ell_j > d] = \Pr[\ell^X > d \text{ and } \ell^Y-j+\min C > d]
  \\ = \Pr[\ell^X > d]\cdot \Pr[\ell^Y-j+\min C > d] \le \exp(-\tfrac{k_1}{r})\cdot \exp(-\tfrac{k_2}{r})
 =\exp(-\tfrac{k_1+k_2}{r}) \le \exp(-\tfrac{k+1}{r}),\end{multline*}
  which completes the proof.
  \end{proof}

  For each $j\in C$, after setting $\ell_j$ in \cref{ln:standard}, the algorithm performs an additional check in \cref{ln:if2}; its implementation is discussed later on.
  If the check succeeds, then the algorithm falls back to computing $\ell_j=\bpLCE_r^{X,Y}(0,j)$,
  using \cref{fct:bplce}, which results in a correct value by definition.
  Otherwise, $\ell_j \ge |X|$, $\ell_j \ge |Y|-j$, or $X(\ell_j-2\Delta\dd \ell_j]\ne Y(j+\ell_j-2\Delta \dd j+\ell_j]$.
  Each of these condition yields $\LCE_{0}^{X,Y}(0,j)\le \ell_j$.
  Hence, due to \cref{clm:standard}, returning $\ell_j$ as $\bpLCE_r^{X,Y}(0,j)$ is then correct.
  This completes the proof that the values returned by \cref{alg:box} are correct.

  However, we still need to describe the implementation of \cref{ln:short} and \cref{ln:if2}.
  The values $\min(\LCE_0^{X,Y}(0,j),2\Delta)$ needed in \cref{ln:short}
  are determined using an auxiliary string
  $T=X[0\dd 2\Delta)\$Y[\min J \dd \max J + 2\Delta)$, where
  $\$$ and each out-of-bounds character does not match any other character.
  The PREF table of $T$, with 
  $\mathrm{PREF}_T[t]=\LCE_0^{T,T}(0,t)$ for $t\in [0\dd |T|]$,
  can be constructed in $\Oh(|T|)=\Oh(|J|)$ time using a textbook algorithm~\cite{Jewels}
  and satisfies $\min(\LCE_0^{X,Y}(0,j),2\Delta)=\mathrm{PREF}_T[2\Delta+1+j-\min J]$
  for each $j\in J$.
  
  Our approach to testing $X(\ell_j-2\Delta\dd \ell_j]= Y(j+\ell_j-2\Delta \dd j+\ell_j]$ in \cref{ln:if2}
  depends on whether $\ell_j = \ell^X$ or not.
  Positions $j\in C$ with $\ell_j = \ell^X$ need to be handled only if $\ell^X<|X|$.
  In this case, $X(\ell_j-2\Delta\dd \ell_j]= Y(j+\ell_j-2\Delta \dd j+\ell_j]$
  holds if and only if $X(\ell^X-2\Delta\dd \ell^X]$ has an occurrence in $T$
  at position $j+\ell^X-2\Delta+1$. 
  Hence, a linear-time pattern matching algorithm is used to 
  find the occurrences of $X(\ell^X-2\Delta\dd \ell^X]$ starting in $T$ between
  positions $\min C + \ell^X -2\Delta+1$ and $\max C + \ell^X - 2\Delta+1$, inclusive.
  Due to $\max C - \min C \le \Delta$, this takes $\Oh(\Delta)$ time.
  Moreover, since $\per(X(\ell^X-2\Delta\dd \ell^X])>\Delta$ holds by \cref{lem:period2},
  \cref{fct:period} implies that there is at most one such occurrence, i.e., at most one position $j$ with $\ell_j = \ell^X$ passes the test in \cref{ln:if2}.

  Next, consider positions $j\in C$ with $\ell_j \ne \ell^X$.
  Since these positions satisfy $\ell_j = \ell^Y-j+\min C$,
  the condition $\ell_j < |Y|-j$ implies $\ell^Y+\min C < |Y|$.
  Moreover, the condition $\ell_j < |X|$ implies $\ell^Y< |X|+j-\min C \le |X|+\max C-\min C$.
  Consequently, such $j\in C$ need to be handled only if 
  $\ell^Y < \min(|Y|-\min C, |X|+\max C-\min C) = |\bar{Y}|$. 
  Our observation is that $X(\ell_j-2\Delta\dd \ell_j]= Y(j+\ell_j-2\Delta \dd j+\ell_j]$
  if and only if $\bar{Y}(\ell^Y-2\Delta \dd \ell^Y]$
  has an occurrence in $X$ at position $\ell^Y-j+\min C-2\Delta+1$.
  Hence, a linear-time pattern matching algorithm is used to 
  find the occurrences of $\bar{Y}(\ell^Y-2\Delta \dd \ell^Y]$ starting in $X$ between
  positions  $\ell^Y-\max C+\min C-2\Delta+1$ and  $\ell^Y-2\Delta+1$, inclusive.
  Due to $\max C - \min C \le \Delta$, this takes $\Oh(\Delta)$ time.
  Moreover, since $\per(\bar{Y}(\ell^Y-2\Delta \dd \ell^Y])>\Delta$ holds by \cref{lem:period2},
  \cref{fct:period} implies that there is at most one such occurrence, i.e., at most one position $j$ with $\ell_j \ne \ell^X$ passes the test in \cref{ln:if2}.

  We conclude that \cref{ln:if2} (across all $j\in C$) can be implemented in $\Oh(\Delta)$
  time and that \cref{ln:special} needs to be executed for at most two indices $j\in C$.
  Consequently, the overall cost of executing Lines~\ref{ln:standard}--\ref{ln:special}
  is $\Ohtilde(\frac1{r}|X|+\Delta)$ with high probability.
  Due to $|\bar{Y}|\le |X|+\Delta$, the cost of calls to \FindBreakTwo\ of \cref{lem:period2}
  is also $\Ohtilde(\frac1{r}|X|+\Delta)$ with high probability.
  As explained above, executing \cref{ln:short} costs $\Oh(|J|)$ time.
  The cost of \cref{ln:only} is $\Ohtilde(\frac{1}{r}|X|)$ with high probability.
  Due to $\Delta = |J|-1$, the total running time is therefore $\Ohtilde(\frac{1}{r}|X|+|J|)$
  with high probability.
  \end{proof}

  We are now ready to prove a counterpart of \cref{prp:batch} for $\bpLCE_r$ queries.
\begin{lemma}\label{lem:batch}
There is a data structure that, initialized with strings $X$ and $Y$, a real parameter $r>0$,
and an integer range $\Delta$,
answers the following queries: given an integer $x$, return $\bpLCE_r^{X,Y}(x,x+\delta)$
for each $\delta \in \Delta$.
The initialization costs $\Ohtilde(\frac{1}{r}|X|)$ time with high probability,
and the queries cost $\Ohtilde(|\Delta|)$ time with high probability.
\end{lemma}
\begin{algorithm}[ht]
  \SetKwComment{Comment}{$\triangleright$\ \rm}{}
  \SetKwFunction{Construction}{Construction}
  \SetKwFunction{Query}{Query}
  
  \Construction{$X$, $Y$, $r$, $\Delta$} \Begin{
    $q := \ceil{r|J|}$\;
    $x := |X|$\;
    \lForEach{$\delta \in \Delta$}{$\ell_{x,\delta} := 0$}
    \While{$x \ge q$}{
      $x := x - q$\;
      $(v_\delta)_{\delta\in \Delta} := \PMTwo{$X[x \dd x+q)$, $Y$, $r$, $\{x+\delta:\delta\in \Delta\}$}$\;
      \ForEach{$\delta \in \Delta$}{
        \lIf{$v_\delta < q$}{$\ell_{x,\delta}:=v_\delta$}
        \lElse{$\ell_{x,\delta} := q + \ell_{x+q,\delta}$}
      }
    }
  }
  \BlankLine

  \Query{$x$} \Begin{
    \lIf{$x\notin [0\dd |X|]$}{\Return{$(0)_{\delta \in\Delta}$}}
    $x' := x + (|X|-x)\bmod q$\;
    \If{$x' \ne x$}{
      $(v_\delta)_{\delta\in \Delta} := \PMTwo{$X[x \dd x')$, $Y$, $r$, $\{x+\delta:\delta\in \Delta\}$}$\;
      \ForEach{$\delta \in \Delta$}{
        \lIf{$v_\delta < x'-x$}{$\ell_{x,\delta}:=v_\delta$}
        \lElse{$\ell_{x,\delta} := x'-x + \ell_{x',\delta}$}
      }
    }
    \Return{$(\ell_{x,\delta})_{\delta \in\Delta}$}\;
  }

  \caption{Implementation of the data structure of \cref{lem:batch}}\label{alg:batch}
  \end{algorithm}
\begin{proof}
Procedures $\Construction{$X$, $Y$, $r$, $\Delta$}$ and $\Query{$x$}$ implementing \cref{lem:batch} are given as \cref{alg:batch}.

The construction algorithm precomputes the answers for all $x\in [0\dd |X|]$
satisfying $x \equiv |X| \pmod{q}$, where $q = \ceil{r|\Delta|}$. 
More formally, for each such $x$, the data structure stores $\ell_{x,\delta}=\bpLCE_{r}^{X,Y}(x,x+\delta)$ for all $\delta\in \Delta$.
First, the values $\ell_{|X|,\delta}$ are set to $0$.
In subsequent iterations, the algorithm computes $\ell_{x,\delta}$ based on $\ell_{x+q,\delta}$.
For this, the procedure \PMTwo{$X[x \dd x+q)$, $Y$, $r$, $\{x+\delta:\delta\in \Delta\}$} of \cref{lem:box}
is called. The resulting values $\bpLCE_r^{X[x\dd x+q),Y}(0,x+\delta)$ are then combined 
with $\ell_{x+q,\delta}=\bpLCE_r^{X[x+q\dd |X|),Y}(0,x+q+\delta)$ based on \cref{lem:combine},
which yields $\bpLCE_r^{X[x\dd |X|),Y}(0,{x+\delta})=\bpLCE_r^{X,Y}(x,{x+\delta})$;
the latter values are stored at $\ell_{x,\delta}$.

The cost of a single iteration is $\Ohtilde(\frac{q}{r}+|\Delta|)=\Ohtilde(\frac{q}{r})$ with high probability,
and the number of iterations is $\Oh(\frac{1}{q}|X|)$, so the total preprocessing time is $\Ohtilde(\frac{1}{r}|X|)$ with high probability.

To answer a query for a given integer $x$, the algorithm needs to compute $\bpLCE_{r}^{X,Y}(x,x+\delta)$ for all $\delta\in \Delta$. If $x\notin [0\dd |X|]$, then these values are equal to $0$ by definition.
Otherwise, the algorithm computes the nearest integer $x'\ge x$ with $x' \equiv |X| \pmod{q}$.
 If $x'=x$, then the sought values have already been precomputed.
Otherwise, the algorithm proceeds based the values $\ell_{x',\delta}$.
For this, the procedure \PMTwo{$X[x \dd x')$, $Y$, $r$, $\{x+\delta:\delta\in \Delta\}$} of \cref{lem:box}
is called.
The resulting values $\bpLCE_r^{X[x\dd x'),Y}(0,x+\delta)$ are then combined 
with $\ell_{x',\delta}=\bpLCE_r^{X[x'\dd |X|),Y}(0,x'+\delta)$ based on \cref{lem:combine},
which yields the sought values $\bpLCE_r^{X[x\dd |X|),Y}(0,x+\delta)=\bpLCE_r^{X,Y}(x,x+\delta)$.

The cost of a query is $\Ohtilde(\frac{x'-x}{r}+|\Delta|)=\Ohtilde(\frac{q}{r}+|\Delta|)=\Ohtilde(\Delta)$ with high probability.
\end{proof}

Finally, we recall that $\bpLCE_r^{X,Y}(x,y)$ satisfies the requirements for $\apLCE{k}^{X,Y}(x,y)$
with probability at least $1-\exp(-\frac{k+1}{r})$.
Consequently, taking sufficiently small $r=\Thetatilde(k+1)$ guarantees success with high probability.
Thus, \cref{lem:batch} yields \cref{prp:batch}, which we restate below.

\prpbatch*

\section{PTAS for Aperiodic Strings}\label{sec:PTAS}

In this section, we design an algorithm distinguishing between $\ED(X,Y)\le k$
and $\ED(X,Y)>(1+\eps)k$, assuming that $X$ does not have a length-$\ell$ 
substring with period at most $2k$. The high-level approach of our solution is
based on the existing algorithms for Ulam distance~\cite{AN10,NSS17}.
The key tool in these algorithms is a method for decomposing $X=X_0\cdots X_m$ and $Y=Y_0\cdots Y_m$ into
short phrases such that $\ED(X,Y)= \sum_{i=0}^m \ED(X_i,Y_i)$ if $\ED(X,Y)\le k$.
While designing such a decomposition in sublinear time for general strings $X$ and $Y$ 
remains a challenging open problem, the lack of long highly periodic substrings makes this task feasible.

\begin{lemma}\label{lem:decompose}
  There exists an algorithm that, given strings $X$ and $Y$, integers $k$ and $\ell$
  such that $\per(X[i\dd \allowbreak i+\ell))>2k$ for each $i\in [0\dd |X|-\ell]$,
  and a real parameter $0<\delta<1$, returns
  factorizations $X=X_0\cdots X_m$ and $Y=Y_0\cdots Y_m$ with
   $m =\Oh(\frac{\delta}{(k+1)\ell}|X|)$ such that
   $|X_i|\le \ceil{\delta^{-1}(k+1)\ell}$ for each $i\in [0\dd m]$ and,
   if $\ED(X,Y)\le k$, then $\Pr[\ED(X,Y)= \sum_{i=0}^m \ED(X_i,Y_i)] \ge 1-\delta$.
   The running time of the algorithm is $\Oh(\frac{\delta}{k+1}|X|)$.
\end{lemma}
\begin{proof}
Let $q = \ceil{\delta^{-1}(k+1)\ell}$. If $|X|\le q$, then the algorithm returns
trivial decompositions of $X$ and $Y$ with $m=0$.
In the following, we assume that $q < |X|$. 
The algorithm chooses $r\in [0\dd q)$ uniformly at random and creates a partition
$X=X_0\cdots X_m$ so that $|X_0|=r$, $|X_i|=q$ for $i\in [1\dd m)$, and $|X_m|\le q$.
This partition clearly satisfies $m =\Oh(\frac{1}{q}|X|) =\Oh(\frac{\delta}{(k+1)\ell}|X|)$ and  $|X_i|\le q=\ceil{\delta^{-1}(k+1)\ell}$ for each $i\in [0\dd m]$.

Let us define $x_i$ for $i\in [0\dd m+1]$ so that $X_i = X[x_{i}\dd x_{i+1})$ for $i\in [0\dd m]$.
For each $i\in [1\dd m]$, the algorithm finds the occurrences of $X[x_i\dd x_i+\ell)$ in $Y$
with starting positions between $x_{i}-k$ and $x_{i}+k$. 
If there is no such occurrence (perhaps due to $|X_i|< \ell$ for $i=m$),
then the algorithm declares a failure and  returns a partition $Y=Y_0\cdots Y_m$
with $Y_0=Y$ and $Y_i=\eps$ for $i\ge 1$.
Otherwise, due to the assumption that $\per(X[x_i\dd x_i+\ell)) > 2k$,
there is exactly one occurrence, say, at position $y_i$.
(Recall that the distance between two positions in $\Occ(P,T)$ is either a period of $P$ or larger than $|P|$.)
The algorithm defines $Y_i=Y[y_i\dd y_{i+1})$ for $i\in [0\dd m]$,
where $y_0 = 0$ and $y_{m+1}=|Y|$.
This approach can be implemented in $\Oh(m\ell)=\Oh(\frac{\delta}{k+1}|X|)$
time using a classic linear-time pattern matching algorithm~\cite{MP70}.

We shall prove that the resulting partition $Y=Y_0\cdots Y_m$ satisfies the requirements.
Assuming that $\ED(X,Y)\le k$, let us fix an optimal alignment between $X$ and $Y$.
We need to prove that, with probability at least $1-\delta$, the fragments $X[x_i\dd x_i+\ell)$
are all matched against  $Y[y_i\dd y_i+\ell)$. By optimality of the alignment, 
this will imply $\Pr[\ED(X,Y)= \sum_{i=0}^m \ED(X_i,Y_i)] \ge 1-\delta$.

We say that a position $x\in [0\dd |X|]$ is an \emph{error}
if $x=|X|$ or (in the alignment considered) the position $X[x]$ is deleted
or matched against a position $Y[y]$ for $y\in [0\dd |Y|)$
such that $Y[y]\ne X[x]$ or $Y[y+1]$ is inserted.
Each edit operation yields at most one error,
so the total number of errors is at most $k+1$.
Moreover, if $[x_i\dd x_i+\ell)$ does not contain any error,
then $X[x_i\dd x_i+\ell)$ is matched exactly against a fragment of $Y$,
and that fragment must be $Y[y_i\dd y_i+\ell)$ (since we considered all starting positions in $[x_i-k\dd x_i+k]$).
Hence, if the algorithm fails, then there is an error $x\in [x_i\dd x_i+\ell)$ for some $i\in [1\dd m]$.
By definition of the decomposition $X=X_0\cdots X_m$, this implies $x\bmod q \in [r\dd r+\ell)\bmod q$, or, 
equivalently, $r\in (x-\ell\dd x]\bmod q$.
The probability of this event is $\frac{\ell}{q}\le \frac{\delta}{k+1}$.
The union bound across all errors yields an upper bound of $\delta$ on the failure probability.
\end{proof}

Next, we present a subroutine that will be applied to individual phrases
of the decompositions of \cref{lem:decompose}. Given that 
the phrases are short, we can afford using the classic Landau--Vishkin algorithm~\cite{LV88}
whenever we find out that the corresponding phrases do not match exactly.
\begin{lemma}\label{lem:simple}
There exists an algorithm that, given strings $X$ and $Y$, and an integer $k\ge 0$,
computes $\ED(X,Y)>0$ exactly, taking $\Ohtilde(|X|+\ED(X,Y)^2)$ time,
or certifies that $\ED(X,Y)\le k$ with high probability,
taking $\Ohtilde(1+\frac{1}{k+1}|X|)$ time with high probability.
\end{lemma}
\begin{proof}
The algorithm first checks if $|X|=|Y|$, and then it samples $X$ with sufficiently large rate $\Thetatilde(\frac{1}{k+1})$ checking whether $X[i]=Y[i]$ for each sampled position $i$.
If the checks succeed, then the algorithm certifies that $\ED(X,Y)\le k$.
This branch takes $\Ohtilde(1+\frac{1}{k+1}|X|)$ time with high probability.
Otherwise,  $\ED(X,Y)>0$, and the algorithm falls back to a procedure of Landau and Vishkin~\cite{LV88}, whose running time is $\Ohtilde(|X|+|Y|+\ED(X,Y)^2)=\Ohtilde(|X|+\ED(X,Y)^2)$. 

As for correctness, it suffices to show that if the checks succeeded, then $\ED(X,Y)\le k$
with high probability. We shall prove a stronger claim that $\HAM(X,Y)\le k$.
For a proof by contradiction, suppose that $\HAM(X,Y)\ge k+1$ and consider some $k+1$ mismatches.
Notice that the sampling rate is sufficiently large that at least one of these mismatches is sampled with high 
probability. This completes the proof.
\end{proof}

The next step is to design a procedure which distinguishes between $\sum_{i=0}^m\ED(X_i,Y_i)\le k$
and $\sum_{i=0}^m\ED(X_i,Y_i)\ge (1+\eps)k$. 
Our approach relies on the Chernoff bound: we apply \cref{lem:simple} to determine $\ED(X_i,Y_i)$
for a small sample of indices $i$, and then we use these values to estimate the sum $\sum_{i=0}^m\ED(X_i,Y_i)$.
\begin{lemma}\label{lem:sum}
There exists an algorithm that, given strings $X_0,\ldots,X_m,Y_0,\ldots,Y_m$, and a real parameter $0<\eps<1$, returns
  YES if $\sum_{i=0}^m\ED(X_i,Y_i)\le k$, and returns
  NO if $\sum_{i=0}^m\ED(X_i,Y_i)\ge (1+\eps)k$.
The algorithm succeeds with high probability, and its running time is $\Ohtilde(qk+k^2+\frac{n}{\eps^2(k+1)})$,
where $q=\max_{i=0}^m |X_i|$ and $n=\sum_{i=0}^m(|X_i|+|Y_i|)$.
\end{lemma}
\begin{proof}
If $k=0$, then the algorithm naively checks if $X_i=Y_i$ for each $i$, which costs $\Oh(n)$ time.
In the following, we assume that $k>0$.

For $i\in [0\dd m]$ and $j\in [0\dd |X_i|+|Y_i|)$, let us define an indicator $r_{i,j}=[\ED(X_i,Y_i)>j]$.
Observe that $\sum_{i=0}^m\ED(X,Y)=\sum_{i=0}^m \sum_{j=0}^{|X_i| + |Y_i|-1} r_{i,j}$.
The algorithm samples independent random variables $R_1,\ldots,R_r$ distributed as a uniformly random among the $n$ terms $r_{i,j}$, where $r=\Thetatilde(\frac{n}{\eps^2 k})$ is sufficiently large,
and returns YES if and only if $\sum_{t=1}^r R_t \le (1+\frac{\eps}{2})\frac{rk}{n}$.

Before we provide implementation details, let us prove the correctness of this approach.
If $\sum_{i=0}^m\ED(X,Y)\le k$, then $\Exp[R_t]\le \frac{k}{n}$. Consequently, the multiplicative Chernoff bound
implies $\Pr[\sum_{t=1}^r R_t \ge (1+\frac{\eps}{2})\frac{rk}{n} ] \le \exp(-\frac{\eps^2rk}{12n})$.
Since $r=\Thetatilde(\frac{n}{\eps^2 k})$ is sufficiently large,
the complementary event holds with high probability.
Similarly, if $\sum_{i=0}^m\ED(X,Y)\ge (1+\eps)k$, then $\Exp[R_t]\ge \frac{(1+\eps)k}{n}$.
Consequently, the multiplicative Chernoff bound implies $\Pr[\sum_{t=1}^r R_t \le (1+\frac{\eps}{2})\frac{rk}{n} ] \le
\Pr[\sum_{t=1}^r R_t \le (1-\frac{\eps}{4})r\frac{(1+\eps)k}{n} ] 
\le \exp(-\frac{\eps^2r(1+\eps)k}{48n})$.
Since $r=\Thetatilde(\frac{n}{\eps^2 k})$ is sufficiently large,
the complementary event holds with high probability.
This finishes the correctness proof.

Evaluating each $R_t$ consists in drawing a term $r_{i_t,j_t}$ uniformly at random and testing
if $r_{i_t,j_t}=1$, that is, whether $\ED(X_{i_t},Y_{i_t})>j_t$. 
For this, the algorithm makes a call to \cref{lem:simple}.
If this call certifies that $\ED(X_{i_t},Y_{i_t})\le j_t$, then the algorithm sets $R_t=0$.
Otherwise, the call returns the exact distance $\ED(X_{i_t},Y_{i_t})>0$, and the algorithm sets $R_t=1$
if and only if $\ED(X_{i_t},Y_{i_t})>j_t$.
Moreover, the algorithm stores the distance $\ED(X_{i_t},Y_{i_t})$ so that if $i_{t'}=i_t$ for some $t'>t$, then
the algorithm uses the stored distance to evaluate $r_{i_{t'},j_{t'}}$ instead of calling \cref{lem:simple} again.
Whenever the sum of the stored distances exceeds $k$, the algorithm terminates and returns NO.
Similarly, a call to \cref{lem:simple} is terminated preemptively (and NO is returned)
if the call takes too much time, indicating that $\ED(X_{i_t},Y_{i_t})>k$.
Consequently, since the function $x\mapsto x^2$ is convex, the total running time of the calls to \cref{lem:simple} that compute $\ED(X_{i_t},Y_{i_t})$
is $\Ohtilde(qk+k^2)$. The total cost of the remaining calls
can be bounded by $\Ohtilde(\sum_{t=1}^r \frac{1}{j_t+1}|X_{i_t}|)$.
Since $\Exp[\frac1{j_t+1}|X_{i_t}| \mid {i_t = i}]\le \ln(|X_{i}|+|Y_i|)+1
\le \ln n + 1$, the total expected running time of these calls is $\Ohtilde(r)=\Ohtilde(\frac{n}{k})$.
When $\sum_{t=1}^r \frac{1}{j_t+1}|X_{i_t}|$ exceeds twice the expectation, the whole algorithm is restarted;
with high probability, the number of restarts is $\Ohtilde(1)$.
\end{proof}

Finally, we obtain the main result of this section by combining \cref{lem:decompose} with \cref{lem:sum}.

\begin{theorem}\label{thm:ptas}
There exists an algorithm that, given strings $X$ and $Y$, integers $k$ and $\ell$
such that $\per(X[i\dd \allowbreak i+\ell))>2k$ for each $i\in [0\dd |X|-\ell]$,
and a real parameter $0<\eps<1$, returns
  YES if $\ED(X,Y)\le k$, and
   NO if $\ED(X,Y)\ge (1+\eps)k$.
With high probability,the algorithm is correct and its running time is $\Ohtilde(\frac{1}{\eps^2(k+1)}|X|+k^2 \ell)$.
\end{theorem}
\begin{proof}
The algorithm performs logarithmically many iterations. 
In each iteration, the algorithm calls \cref{lem:decompose} (with $\delta=\frac12$) to obtain decompositions
$X=X_0\cdots X_m$ and $Y=Y_0\cdots Y_m$.
Then, the phrases are processed using \cref{lem:sum}.
If this subroutine returns YES, then the algorithm also returns YES, because
$\ED(X,Y)\le \sum_{i=0}^{m} \ED(X_i,Y_i)<(1+\eps)k$ with high probability.
On the other hand, if each call to \cref{lem:sum} returns NO, then the algorithm returns NO.
If $\ED(X,Y)\le k$, then with high probability, $\ED(X,Y)= \sum_{i=0}^{m} \ED(X_i,Y_i)$ holds in at least one iteration; thus,  $\ED(X,Y)>k$ holds with high probability if the algorithm returns NO.

As for the running time, the calls to \cref{lem:decompose} cost $\Ohtilde(\frac{1}{k+1}|X|)$ time,
and the calls to \cref{lem:sum} cost $\Ohtilde((k+1)^2\ell +\frac{1}{\eps^2(k+1)}(|X|+|Y|))$ time
because $\max |X_i| = \Oh((k+1)\ell)$.
\end{proof}

\section{Random Walk over Samples}\label{sec:randomwalk}
In this section, we describe the sampled random walk process. This is used in \cref{sec:embed} to embed edit distance to Hamming distance in sublinear time.

\SetKwFunction{SampledRandomWalk}{SampledRandomWalk}
\begin{algorithm}
\renewcommand{\gets}{:=}
	\SetKwComment{Comment}{$\triangleright$\ \rm}{}
\caption{\protect\SampledRandomWalk{$X$, $Y$, $k$, $p$}}\label{alg:random-walk}
$x \gets 0$, $y \gets 0$, $c \gets 0$\Comment*{Initialization}
\While{$x < |X|$ \KwSty{and} $y < |Y|$}{\label{ln:while}
		Let $s \sim \textrm{Bin}(1,\frac{2\ln n}{p})$\Comment*{Biased coin}
		\If{$s=1$ \KwSty{and} $X[x]\neq Y[y]$}{\label{ln:if}
			Let $r \sim \textrm{Bin}(1,\frac12)$\Comment*{Unbiased coin}\label{ln:r}
			$x \gets x + r$\;
			$y \gets y + (1-r)$\;
			$c \gets c + 1$\;\label{ln:cost}
		}\Else{
			$x \gets x+1$\;\label{ln:x1}
			$y \gets y+1$\;\label{ln:y1}
		}
}
\Return{$c+\max(|X|-x,|Y|-y)\le 1296k^2$}\;
\end{algorithm}

Given strings $X,Y \in \Sigma^{\le n}$ and integer parameters $k \geq 0$, $p \ge 2\ln n$, 
\cref{alg:random-walk} scans $X$ and $Y$ from left to right.
The currently processed positions are denoted by $x$ and $y$,
respectively. At each iteration, the algorithm tosses a biased coin: with probability $\frac{2\ln n}{p}$,
it chooses to compare $X[x]$ with $Y[y]$ and, in case of a mismatch ($X[x]\ne Y[y]$),
it tosses an unbiased coin to decide whether to increment $x$ or $y$.
In the remaining cases, both $x$ and $y$ are incremented.
Once the algorithm completes scanning $X$ or $Y$,
it returns YES or NO depending on whether 
the number of mismatches encountered is at most $1296k^2-\max(|X|-x,|Y|-y)$.

\begin{theorem}\label{thm:random-walk}
Given strings $X, Y \in \Sigma^{\le n}$ and integers $k \geq 0$ and $2\ln n \le p \le n$, \cref{alg:random-walk} returns YES with probability at least $\frac23$ if $\ED(X,Y) \le k$,
and NO with probability at least $1-\frac{1}{n}$ if $\ED(X,Y) \ge (1296k^2+1)p$.
Moreover, \cref{alg:random-walk} can be implemented in $\Ohtilde(\frac{n}{p})$ time. 
\end{theorem}
\paragraph*{YES Case}
Recall that the \emph{indel distance} $\IDD(\cdot,\cdot)$ is defined so
that $\IDD(X,Y)$ is the minimum number of character insertions and deletions needed to transform $X$ to $Y$
(the cost of a character substitution is 2 in this setting),
and observe that $\ED(X,Y) \le \IDD(X,Y)\le 2\ED(X,Y)$.

We analyze how $D:=\IDD(X[x\dd |X|),Y[y\dd |Y|))$ changes throughout the execution of \cref{alg:random-walk}.
Let $D_0$ be the initial value of $D$ and let $D_i$ be the value of $D$ after the $i$th iteration
of the algorithm, where $i\in [1\dd t]$ and $t$ is the total number of iterations.
We say that iteration $i$ is a \emph{mismatch iteration} if the condition in \cref{ln:if} is satisfied.
The following lemma gathers properties of the values $D_0,D_1,\ldots,D_t$:
\begin{lemma}\label{lem:d}
We have $D_0 = \IDD(X,Y)\le 2\ED(X,Y)$. Moreover, the following holds for each iteration $i\in [1\dd t]$:
\begin{enumerate}[label={\rm(\alph*)}]
	\item\label{it:lb} $D_i$ is a non-negative integer,
	\item\label{it:nmis} If $i$ is not a mismatch iteration, then $D_i \le D_{i-1}$.

	\item\label{it:mis} If $i$ is a mismatch iteration, then $D_i = D_{i-1}-1$ or $D_i = D_{i-1}+1$,
	and $D_i = D_{i-1}-1$ holds for at least one of the two possible outcomes of $r$ in \cref{ln:r}.
	\item\label{it:0} If $D_{i-1}=0$, then $D_{i}=0$.
\end{enumerate}
\end{lemma}
\begin{proof}
Property~\ref{it:lb} is clear from the definition of $D$.

As for Property~\ref{it:nmis}, let us consider an optimal indel-distance alignment resulting in $D_{i-1}=\IDD(X[x\dd |X|),Y[y\dd |Y|))$
at the beginning of iteration $i$, and transform it to an alignment 
between $X[x+1\dd |X|)$ and $Y[y+1\dd |Y|)$ by discarding $X[x]$ and $Y[y]$
and deleting characters matched with $X[x]$ or $Y[y]$, if any.
Note that a character $X[x']$ with $x'>x$ is deleted only 
if the original alignment deletes $Y[y]$, and a character $Y[y']$ with $y'>y$ is deleted only if the original alignment deletes $X[x]$.
Hence, the alignment cost does not increase and $D_{i} \le D_{i-1}$.

As for Property~\ref{it:mis}, observe that incrementing $x$ or $y$ changes the value of $D$ by exactly~1.
Since $X[x]\ne Y[y]$ holds at the beginning of iteration $i$, every optimal alignment between $X[x\dd |X|)$ and $Y[y\dd |Y|)$ deletes $X[x]$ or $Y[y]$. Incrementing $x$ or $y$, respectively,
then results in $D_{i}=D_{i-1}-1$.

As for Property~\ref{it:0}, we note that once $X[x\dd |X|)=Y[y\dd |Y|)$,
no subsequent iteration will be a mismatch iteration.
\end{proof}

Now, consider a 1-dimensional random walk $(W_0)_{i\ge 0}$ that starts with $W_0 = 2k$
and moves 1 unit up or down at every step with equal probability $\frac12$. 
Let us couple this random walk
with the execution of \cref{alg:random-walk}.
Let $i_1,\ldots,i_c$ be the mismatch iterations.
For each $j\in [1\dd c]$ such that exactly one choice of $r$ at iteration $i_{j}$ results
in decrementing $D$, we require that $W_{j}=W_{j-1}-1$ if and only if $D_{i_j}=D_{i_j-1}-1$.
Otherwise, we keep $W_j-W_{j-1}$ independent of the execution.
As each coin toss in \cref{alg:random-walk} uses fresh randomness, 
the steps of the random walk remain unbiased and independent from each other.

Now, \cref{lem:d} implies that $D_0 \le W_0$, that  $D_{i_j} \le W_j$ holds for $j\in [1\dd c]$,
and that $D_t\le W_c$ holds upon termination of \cref{alg:random-walk}.
In particular, the hitting time $T = \min\{j : W_j = 0\}$ satisfies $T \ge c+W_c \ge c+D_t$.
However, as proved in~\cite[Theorem 2.17]{Levin2017}, $\Pr[T \le N]\ge 1-\frac{12k}{\sqrt{N}}$.
Thus, $\Pr[c+D_t \le 1296k^2]\ge 1-\frac{12k}{\sqrt{1296k^2}}=\frac23$.
Since $D_t= \max(|X|-x,|Y|-y)$,
this completes the proof of the YES case.

\paragraph*{NO Case}
If $s=0$ and $X[x]\ne Y[y]$ holds at the beginning of some iteration of \cref{alg:random-walk}, we say the algorithm \emph{misses} the mismatch between $X[x]$ and $Y[y]$.
Let us bound probability of missing many mismatches in a row.

\begin{lemma}\label{lem:miss}
Consider the values $x,y$ at the beginning of iteration $i$ of \cref{alg:random-walk}.
Conditioned on any random choices made prior to iteration $i$,
the probability that \cref{alg:random-walk} misses the leftmost $p$ mismatches between
$X[x\dd |X|)$ and $Y[y\dd |Y|)$ is at most $n^{-2}$. 
\end{lemma}
\begin{proof}
Prior to detecting any mismatch between $X[x\dd |X|)$ and $Y[y\dd |Y|)$,
the algorithm scans these strings from left to right, comparing the characters
at positions sampled independently with rate $\frac{2\ln n}{p}$.
Hence, the probability of missing the first $p$ mismatches is $(1-\frac{2\ln n}{p})^p
\le n^{-2}$.
\end{proof}

Now, observe that an execution of \cref{alg:random-walk}
yields an edit-distance alignment between $X$ and $Y$:
consider values of $x$ and $y$ at an iteration $i$ of the algorithm.
If $i$ is a mismatch iteration, then $X[x]$ or $Y[y]$ is deleted depending on whether the algorithm increments $x$ or $y$. Otherwise, $X[x]$ is aligned against $Y[y]$ (which might be a substitution).
Finally, all $\max(|X|-x,|Y|-y)$ characters remaining in $X[x\dd |X|)$ or $Y[y\dd |Y|)$ after the last iteration $t$ are deleted.
The total number of deletions is thus $c+\max(|X|-x,|Y|-y)$,
and every substitution corresponds to a missed mismatch.
By \cref{lem:miss},
for every block of subsequent non-mismatch iterations, with probability at least $1-\frac{1}{n^2}$,
there are at most $p-1$ missed mismatches.
Overall, with probability at least $1-\frac{1}{n}$, there are at most $(p-1)\allowbreak(c+1)$ missed mismatches,
and $\ED(X,Y) < {(c+\max(|X|-x,|Y|-y)+1)p}$. Hence, the algorithm returns NO if $\ED(X,Y)\ge (1296k^2+1)p$.

\paragraph*{Efficient Implementation}
Finally, we observe that iterations with $s=0$ do not need to be executed explicitly:
it suffices to repeat the following process: draw (from a geometric distribution $\mathrm{Geo}(0,\frac{2\ln n}{p})$)
the number $\delta$ of subsequent iterations with $s=0$, increase both $x$ and $y$ by $\delta$,
and then execute a single iteration with $s=1$.
The total number of iterations is at most $|X|+|Y|\le 2n$, 
and the number of iterations with $s=1$ is $\Oh(\frac{n\ln n}{p})$ with high probability.

This completes the proof of Theorem~\ref{thm:random-walk}.\qed%
 
 \section{Embedding Edit Distance to Hamming Distance}\label{sec:embed}

A randomized embedding of edit distance to Hamming distance is given by a function $f$ such that
$\HAM(f(X,R), f(Y,R)) \approx \ED(X,Y)$ holds with good probability over the randomness $R$
for any two strings $X$ and $Y$.  
 Chakraborty, Goldenberg, and Koucký~\cite{CGK16} gave one such randomized embedding:
 \begin{theorem*}[{\cite[Theorem 1]{CGK16}}]
 For every integer $n \ge 1$, there is an integer $\ell=\Oh(\log{n})$ and a function $f: \{\mathtt{0},\mathtt{1}\}^n \times \{\mathtt{0},\mathtt{1}\}^{\ell} \rightarrow \{\mathtt{0},\mathtt{1}\}^{3n}$ such that, for every $X, Y \in \{\mathtt{0},\mathtt{1}\}^n$,
 \begin{align*}
 \tfrac{1}{2}\ED(X,Y) \leq \HAM(f(X,R), f(Y,R)) \leq \Oh(\ED^2(X,Y))
 \end{align*}
 holds with probability at least $\frac{2}{3}-e^{-\Omega(n)}$ over a uniformly random choice of $R \in \{\mathtt{0},\mathtt{1}\}^{\ell}$. Moreover, $f$ can be evaluated in linear time.
 \end{theorem*}

Their algorithm utilizes $3n$ hash functions $h_1,h_2,\ldots,h_{3n}$ mapping $\{\mathtt{0},\mathtt{1}\}$ to $\{0,1\}$. It scans $X$ sequentially, and if it is at $X[x]$ in iteration $i$, it appends $X[x]$ to the embedding and increments $x$ by $h_i(X[x])$. The latter can be viewed as tossing an unbiased coin and, depending on its outcome, either staying at $X[x]$ or moving to $X[{x+1}]$. The algorithm uses $\Oh(n)$ random bits, which can be reduced to $\Oh(\log n)$
using Nisan's pseudorandom number generator~\cite{Nisan1992}.

By utilizing random walk over samples, we provide the first sublinear-time randomized embedding from edit to Hamming distance. 
Given a parameter $p\ge 2\ln n$, we sample $S\sub [1\dd 3n]$, with each index $i\in [1\dd 3n]$ contained in $S$ independently with probability $\frac{2\ln n}{p}$.
We then independently draw $|S|$ uniformly random bijections $h_1, h_2,\ldots,h_{|S|} : \{\mathtt{0},\mathtt{1}\}\to \{0,1\}$.
The shared randomness $R$ consists of $S$ and $h_1, h_2,\ldots, h_{|S|}$. 
We prove the following theorem.

\begin{theorem}\label{thm:embed}
 For every integer $n \ge 1$ and $p \ge 2\ln n$, there is an integer $\ell=\Oh(\log n)$ and a function $f: \{\mathtt{0},\mathtt{1}\}^n \times \{\mathtt{0},\mathtt{1}\}^{\ell} \rightarrow \{\mathtt{0},\mathtt{1}\}^{\Ohtilde(\frac{n}{p})}$ such that, for every $X, Y \in \{\mathtt{0},\mathtt{1}\}^n$,
 \[
 \tfrac{\ED(X,Y)-p+1}{p+1} \leq \HAM(f(X,R), f(Y,R)) \leq \Oh(\ED^2(X,Y))
\]
holds with probability at least $\frac{2}{3}-n^{-\Omega(1)}$ over a uniformly random choice of $R \in \{\mathtt{0},\mathtt{1}\}^{\ell}$. Moreover, $f$ can be evaluated in $\Ohtilde(\frac{n}{p})$ time.
 \end{theorem}

 \cref{alg:embedding} provides the pseudocode of the embedding.

 \begin{algorithm}
	\SetKwFunction{SublinearEmbedding}{SublinearEmbedding}
	\caption{$\textsf{SublinearEmbedding}$($X, \langle S, h_1,\ldots,h_{|S|}\rangle$)}\label{alg:embedding}
	$Output :=\eps$, $j:=0$, $x := 0$, $X' := X\cdot \mathtt{0}^{3n}$\;\label{ln:init}
	\For{$i:=1$ \KwSty{to} $3n$}{
		\If{$i\in S$}{
      $Output[j] := X'[x]$\;
      $j := j+1$\;
			$x := x + h_{j}(X'[x])$\;
		}\lElse{$x := x + 1$}
	}
	\Return{$Output$}
	\end{algorithm}

\paragraph{Interpretation through \cref{alg:random-walk}} 
Consider running \cref{alg:random-walk} for $X'=X\mathtt{0}^{3n}$ and $Y'=Y\mathtt{0}^{3n}$,
modified as follows: if $s = 1$ and $X'[x]=Y'[y]$,
then we still toss the unbiased coin $r$ and, depending on the result, either increment both $x$ and $y$,
or none of them. Observe that this has no impact of the outcome of processing $(x,y)$:
ultimately, both $x$ and $y$ are incremented and the cost $c$ remains unchanged.
Now, let us further modify \cref{alg:random-walk} so that its execution uses $R$ as the source of randomness:
we use the event $\{i\in S\}$ for the $i$th toss of the biased coin~$s$
and the value $h_j(X'[x])$ for the $j$th toss of the unbiased coin~$r$
(which is now tossed whenever $s=1$). 
If $S$ was drawn $[1\dd \infty)$, this
would perfectly implement the coins. 
However, $S\sub [1\dd 3n]$, so the transformation is valid only for the first $3n$ iterations. 
Nevertheless, for each iteration, the probability of incrementing 
$x$ is $1-\frac{\ln n}{p}\ge \frac12$, so $x\ge n$ and, symmetrically, $y\ge n$ hold with high probability after iteration $3n$,
and then \cref{alg:random-walk} cannot detect any mismatch.
Hence, \cref{alg:embedding} with high probability simulates the treatment of $X'$ and $Y'$ by \cref{alg:random-walk}.
Moreover, $c$ in \cref{alg:random-walk}
corresponds to the number of iterations with $s=1$ and $X'[x]\ne Y'[y]$,
and this number of iterations is precisely $\HAM(f(X,R),f(Y,R))$.

\paragraph*{YES Case}
\cref{alg:random-walk} with $k=\ED(X,Y) = \ED(X',Y')$ returns YES with probability at least $\frac23$.
Thus, $\HAM(f(X,R),f(Y,R)) = c \le 1296k^2 = \Oh(\ED(X,Y)^2)$ with probability at least $\frac23$.

\paragraph*{NO Case} 
As proved in \cref{sec:randomwalk}, $\ED(X',Y')\le c+\max(|X'|-x,|Y'|-y)+(p-1)(c+1)$
holds with probability at least $1-\frac1n$. Due to $|X'|=|Y'|$ and $|x-y|\le c$, we
deduce $\ED(X,Y)=\ED(X',Y')\le 2c+(p-1)(c+1)=(p-1)+(p+1)c$, i.e., $\HAM(f(X,R),f(Y,R))= c \ge \frac{\ED(X,Y)-p+1}{p+1}$.

\paragraph*{Efficient implementation}
To complete the proof, we note that \cref{alg:embedding} can be implemented in $\Oh(|S|)$
time by batching the iterations $i$ with $i\notin S$ (for each such iteration, it suffices to increment~$i$ and~$x$).
Moreover, $|S|=\Oh(\frac{n\ln n}{p})$ holds with high probability.\qed

\bibliographystyle{plainurl}
\bibliography{references}

\begin{thebibliography}{10}

\bibitem{ADGIR03}
Alexandr Andoni, Michel Deza, Anupam Gupta, Piotr Indyk, and Sofya
  Raskhodnikova.
\newblock Lower bounds for embedding edit distance into normed spaces.
\newblock In {\em 14th Annual {ACM-SIAM} Symposium on Discrete Algorithms, SODA
  2003}, pages 523--526, 2003.
\newblock URL: \url{http://dl.acm.org/citation.cfm?id=644108.644196}.

\bibitem{AKO10}
Alexandr Andoni, Robert Krauthgamer, and Krzysztof Onak.
\newblock Polylogarithmic approximation for edit distance and the asymmetric
  query complexity.
\newblock In {\em 51st Annual {IEEE} Symposium on Foundations of Computer
  Science, {FOCS} 2010}, pages 377--386. {IEEE}, 2010.
\newblock \href {https://doi.org/10.1109/FOCS.2010.43}
  {\path{doi:10.1109/FOCS.2010.43}}.

\bibitem{AN10}
Alexandr Andoni and Huy~L. Nguyen.
\newblock Near-optimal sublinear time algorithms for ulam distance.
\newblock In Moses Charikar, editor, {\em 21st Annual {ACM-SIAM} Symposium on
  Discrete Algorithms, {SODA} 2010}, pages 76--86. {SIAM}, 2010.
\newblock \href {https://doi.org/10.1137/1.9781611973075.8}
  {\path{doi:10.1137/1.9781611973075.8}}.

\bibitem{AN20}
Alexandr Andoni and Negev~Shekel Nosatzki.
\newblock Edit distance in near-linear time: it's a constant factor.
\newblock In Sany Irani, editor, {\em 61st Annual {IEEE} Symposium on
  Foundations of Computer Science, {FOCS} 2020}. IEEE, 2020.
\newblock \href {https://doi.org/10.1109/FOCS46700.2020.00096}
  {\path{doi:10.1109/FOCS46700.2020.00096}}.

\bibitem{AO12}
Alexandr Andoni and Krzysztof Onak.
\newblock Approximating edit distance in near-linear time.
\newblock {\em {SIAM} Journal on Computing}, 41(6):1635--1648, 2012.
\newblock \href {https://doi.org/10.1137/090767182}
  {\path{doi:10.1137/090767182}}.

\bibitem{BI18}
Arturs Backurs and Piotr Indyk.
\newblock Edit distance cannot be computed in strongly subquadratic time
  (unless {SETH} is false).
\newblock {\em {SIAM} Journal on Computing}, 47(3):1087--1097, 2018.
\newblock \href {https://doi.org/10.1137/15M1053128}
  {\path{doi:10.1137/15M1053128}}.

\bibitem{BJKK04}
Ziv Bar{-}Yossef, T.~S. Jayram, Robert Krauthgamer, and Ravi Kumar.
\newblock Approximating edit distance efficiently.
\newblock In {\em 45th Annual {IEEE} Symposium on Foundations of Computer
  Science, {FOCS} 2004}, pages 550--559. {IEEE}, 2004.
\newblock \href {https://doi.org/10.1109/FOCS.2004.14}
  {\path{doi:10.1109/FOCS.2004.14}}.

\bibitem{BEKMRRS03}
Tugkan Batu, Funda Erg{\"{u}}n, Joe Kilian, Avner Magen, Sofya Raskhodnikova,
  Ronitt Rubinfeld, and Rahul Sami.
\newblock A sublinear algorithm for weakly approximating edit distance.
\newblock In Lawrence~L. Larmore and Michel~X. Goemans, editors, {\em 35th
  Annual {ACM} Symposium on Theory of Computing, STOC 2003}, pages 316--324.
  {ACM}, 2003.
\newblock \href {https://doi.org/10.1145/780542.780590}
  {\path{doi:10.1145/780542.780590}}.

\bibitem{BES06}
Tugkan Batu, Funda Erg{\"{u}}n, and S{\"{u}}leyman~Cenk Sahinalp.
\newblock Oblivious string embeddings and edit distance approximations.
\newblock In {\em 17th Annual {ACM-SIAM} Symposium on Discrete Algorithms,
  {SODA} 2006}, pages 792--801. {ACM} Press, 2006.
\newblock URL: \url{http://dl.acm.org/citation.cfm?id=1109557.1109644}.

\bibitem{BZ16}
Djamal Belazzougui and Qin Zhang.
\newblock Edit distance: Sketching, streaming, and document exchange.
\newblock In Irit Dinur, editor, {\em 57th Annual {IEEE} Symposium on
  Foundations of Computer Science, {FOCS} 2016}, pages 51--60. {IEEE}, 2016.
\newblock \href {https://doi.org/10.1109/FOCS.2016.15}
  {\path{doi:10.1109/FOCS.2016.15}}.

\bibitem{BEGHS18}
Mahdi Boroujeni, Soheil Ehsani, Mohammad Ghodsi, Mohammad~Taghi Hajiaghayi, and
  Saeed Seddighin.
\newblock Approximating edit distance in truly subquadratic time: Quantum and
  mapreduce.
\newblock In Artur Czumaj, editor, {\em 29th Annual {ACM-SIAM} Symposium on
  Discrete Algorithms, {SODA} 2018}, pages 1170--1189. {SIAM}, 2018.
\newblock \href {https://doi.org/10.1137/1.9781611975031.76}
  {\path{doi:10.1137/1.9781611975031.76}}.

\bibitem{BCR20}
Joshua Brakensiek, Moses Charikar, and Aviad Rubinstein.
\newblock A simple sublinear algorithm for gap edit distance, 2020.
\newblock \href {http://arxiv.org/abs/2007.14368} {\path{arXiv:2007.14368}}.

\bibitem{BGZ17}
Joshua Brakensiek, Venkatesan Guruswami, and Samuel Zbarsky.
\newblock Efficient low-redundancy codes for correcting multiple deletions.
\newblock {\em {IEEE} Transactions on Information Theory}, 64(5):3403--3410,
  2018.
\newblock \href {https://doi.org/10.1109/TIT.2017.2746566}
  {\path{doi:10.1109/TIT.2017.2746566}}.

\bibitem{BR20}
Joshua Brakensiek and Aviad Rubinstein.
\newblock Constant-factor approximation of near-linear edit distance in
  near-linear time.
\newblock In Konstantin Makarychev, Yury Makarychev, Madhur Tulsiani, Gautam
  Kamath, and Julia Chuzhoy, editors, {\em 52nd Annual {ACM} Symposium on
  Theory of Computing, {STOC} 2020}, pages 685--698. {ACM}, 2020.
\newblock \href {https://doi.org/10.1145/3357713.3384282}
  {\path{doi:10.1145/3357713.3384282}}.

\bibitem{BG95}
Dany Breslauer and Zvi Galil.
\newblock Finding all periods and initial palindromes of a string in parallel.
\newblock {\em Algorithmica}, 14(4):355--366, 1995.
\newblock \href {https://doi.org/10.1007/BF01294132}
  {\path{doi:10.1007/BF01294132}}.

\bibitem{CDGKS18}
Diptarka Chakraborty, Debarati Das, Elazar Goldenberg, Michal Kouck{\'{y}}, and
  Michael~E. Saks.
\newblock Approximating edit distance within constant factor in truly
  sub-quadratic time.
\newblock In Mikkel Thorup, editor, {\em 59th Annual {IEEE} Symposium on
  Foundations of Computer Science, {FOCS} 2018}, pages 979--990. {IEEE}, 2018.
\newblock \href {https://doi.org/10.1109/FOCS.2018.00096}
  {\path{doi:10.1109/FOCS.2018.00096}}.

\bibitem{CGK16}
Diptarka Chakraborty, Elazar Goldenberg, and Michal Kouck{\'{y}}.
\newblock Streaming algorithms for embedding and computing edit distance in the
  low distance regime.
\newblock In Daniel Wichs and Yishay Mansour, editors, {\em 48th Annual {ACM}
  Symposium on Theory of Computing, {STOC} 2016}, pages 712--725. {ACM}, 2016.
\newblock \href {https://doi.org/10.1145/2897518.2897577}
  {\path{doi:10.1145/2897518.2897577}}.

\bibitem{CGKK18}
Moses Charikar, Ofir Geri, Michael~P. Kim, and William Kuszmaul.
\newblock On estimating edit distance: Alignment, dimension reduction, and
  embeddings.
\newblock In Ioannis Chatzigiannakis, Christos Kaklamanis, D{\'{a}}niel Marx,
  and Donald Sannella, editors, {\em 45th International Colloquium on Automata,
  Languages, and Programming, {ICALP} 2018}, volume 107 of {\em LIPIcs}, pages
  34:1--34:14. Schloss Dagstuhl--Leibniz-Zentrum f{\"{u}}r Informatik, 2018.
\newblock \href {https://doi.org/10.4230/LIPIcs.ICALP.2018.34}
  {\path{doi:10.4230/LIPIcs.ICALP.2018.34}}.

\bibitem{Jewels}
Maxime Crochemore and Wojciech Rytter.
\newblock {\em Jewels of Stringology}.
\newblock World Scientific, 2003.
\newblock \href {https://doi.org/10.1142/4838} {\path{doi:10.1142/4838}}.

\bibitem{DKFPOS}
Arthur~L. Delcher, Simon Kasif, Robert~D. Fleischmann, Jeremy Peterson, Owen
  White, and Steven~L. Salzberg.
\newblock Alignment of whole genomes.
\newblock {\em Nucleic Acids Research}, 27(11):2369--2376, 1999.
\newblock \href {https://doi.org/10.1093/nar/27.11.2369}
  {\path{doi:10.1093/nar/27.11.2369}}.

\bibitem{FW65}
Nathan~J. Fine and Herbert~S. Wilf.
\newblock Uniqueness theorems for periodic functions.
\newblock {\em Proceedings of the American Mathematical Society},
  16(1):109--114, 1965.
\newblock \href {https://doi.org/10.1090/S0002-9939-1965-0174934-9}
  {\path{doi:10.1090/S0002-9939-1965-0174934-9}}.

\bibitem{GKS19}
Elazar Goldenberg, Robert Krauthgamer, and Barna Saha.
\newblock Sublinear algorithms for gap edit distance.
\newblock In David Zuckerman, editor, {\em 60th Annual {IEEE} Symposium on
  Foundations of Computer Science, {FOCS} 2019}, pages 1101--1120. {IEEE},
  2019.
\newblock \href {https://doi.org/10.1109/FOCS.2019.00070}
  {\path{doi:10.1109/FOCS.2019.00070}}.

\bibitem{GRS20}
Elazar Goldenberg, Aviad Rubinstein, and Barna Saha.
\newblock Does preprocessing help in fast sequence comparisons?
\newblock In Konstantin Makarychev, Yury Makarychev, Madhur Tulsiani, Gautam
  Kamath, and Julia Chuzhoy, editors, {\em 52nd Annual {ACM} Symposium on
  Theory of Computing, {STOC} 2020}, pages 657--670. {ACM}, 2020.
\newblock \href {https://doi.org/10.1145/3357713.3384300}
  {\path{doi:10.1145/3357713.3384300}}.

\bibitem{H19}
Bernhard Haeupler.
\newblock Optimal document exchange and new codes for insertions and deletions.
\newblock In David Zuckerman, editor, {\em 60th Annual {IEEE} Symposium on
  Foundations of Computer Science, {FOCS} 2019}, pages 334--347. {IEEE}, 2019.
\newblock \href {https://doi.org/10.1109/FOCS.2019.00029}
  {\path{doi:10.1109/FOCS.2019.00029}}.

\bibitem{IP01}
Russell Impagliazzo and Ramamohan Paturi.
\newblock On the complexity of $k$-{SAT}.
\newblock {\em Journal of Computer and System Sciences}, 62(2):367--375, 2001.
\newblock \href {https://doi.org/10.1006/jcss.2000.1727}
  {\path{doi:10.1006/jcss.2000.1727}}.

\bibitem{KS20}
Michal Kouck{\'{y}} and Michael~E. Saks.
\newblock Constant factor approximations to edit distance on far input pairs in
  nearly linear time.
\newblock In Konstantin Makarychev, Yury Makarychev, Madhur Tulsiani, Gautam
  Kamath, and Julia Chuzhoy, editors, {\em 52nd Annual {ACM} Symposium on
  Theory of Computing, {STOC} 2020}, pages 699--712. {ACM}, 2020.
\newblock \href {https://doi.org/10.1145/3357713.3384307}
  {\path{doi:10.1145/3357713.3384307}}.

\bibitem{LMS98}
Gad~M. Landau, Eugene~W. Myers, and Jeanette~P. Schmidt.
\newblock Incremental string comparison.
\newblock {\em {SIAM} Journal on Computing}, 27(2):557--582, 1998.
\newblock \href {https://doi.org/10.1137/S0097539794264810}
  {\path{doi:10.1137/S0097539794264810}}.

\bibitem{LV88}
Gad~M. Landau and Uzi Vishkin.
\newblock Fast string matching with $k$ differences.
\newblock {\em Journal of Computer and System Sciences}, 37(1):63--78, 1988.
\newblock \href {https://doi.org/10.1016/0022-0000(88)90045-1}
  {\path{doi:10.1016/0022-0000(88)90045-1}}.

\bibitem{Lev65}
Vladimir~I. Levenshtein.
\newblock Binary codes capable of correcting deletions, insertions and
  reversals.
\newblock {\em Doklady Akademii Nauk SSSR}, 163(4):845--848, 1965.
\newblock URL: \url{http://mi.mathnet.ru/eng/dan31411}.

\bibitem{Levin2017}
David Levin and Yuval Peres.
\newblock {\em Markov Chains and Mixing Times}.
\newblock American Mathematical Society, 2017.
\newblock \href {https://doi.org/10.1090/mbk/107} {\path{doi:10.1090/mbk/107}}.

\bibitem{MP70}
James~H. Morris, Jr. and Vaughan~R. Pratt.
\newblock A linear pattern-matching algorithm.
\newblock Technical Report~40, Department of Computer Science, University of
  California, Berkeley, 1970.

\bibitem{NSS17}
Timothy Naumovitz, Michael~E. Saks, and C.~Seshadhri.
\newblock Accurate and nearly optimal sublinear approximations to ulam
  distance.
\newblock In Philip~N. Klein, editor, {\em 28th Annual {ACM-SIAM} Symposium on
  Discrete Algorithms, {SODA} 2017}, pages 2012--2031. {SIAM}, 2017.
\newblock \href {https://doi.org/10.1137/1.9781611974782.131}
  {\path{doi:10.1137/1.9781611974782.131}}.

\bibitem{Nisan1992}
Noam Nisan.
\newblock Pseudorandom generators for space-bounded computation.
\newblock {\em Combinatorica}, 12(4):449--461, 1992.
\newblock \href {https://doi.org/10.1007/BF01305237}
  {\path{doi:10.1007/BF01305237}}.

\bibitem{OR07}
Rafail Ostrovsky and Yuval Rabani.
\newblock Low distortion embeddings for edit distance.
\newblock {\em Journal of the {ACM}}, 54(5):23, 2007.
\newblock \href {https://doi.org/10.1145/1284320.1284322}
  {\path{doi:10.1145/1284320.1284322}}.

\bibitem{S14}
Barna Saha.
\newblock The {D}yck language edit distance problem in near-linear time.
\newblock In {\em 55th Annual {IEEE} Symposium on Foundations of Computer
  Science, {FOCS} 2014}, pages 611--620. {IEEE}, 2014.
\newblock \href {https://doi.org/10.1109/FOCS.2014.71}
  {\path{doi:10.1109/FOCS.2014.71}}.

\bibitem{ZZ17}
Haoyu Zhang and Qin Zhang.
\newblock Embedjoin: Efficient edit similarity joins via embeddings.
\newblock In {\em 23rd {ACM} {SIGKDD} International Conference on Knowledge
  Discovery and Data Mining, KDD 2017}, pages 585--594. {ACM}, 2017.
\newblock \href {https://doi.org/10.1145/3097983.3098003}
  {\path{doi:10.1145/3097983.3098003}}.

\end{thebibliography}

\end{document}